%% file: arxiv.tex
\documentclass[aps,11pt,notitlepage,onecolumn,notitlepage,superscriptaddress,tightenlines,nofootinbib,secnumarabic]{revtex4-2}
\pdfoutput=1

\include{def}

\newcommand{\deff}[1]{\textbf{\emph{#1}}}

\begin{document}

\title{Functional analytic insights into irreversibility of quantum resources}

\author{Bartosz Regula}
\email{bartosz.regula@gmail.com}
\affiliation{Mathematical Quantum Information RIKEN Hakubi Research Team, RIKEN Cluster for Pioneering Research (CPR) and RIKEN Center for Quantum Computing (RQC), Wako, Saitama 351-0198, Japan}
\affiliation{Department of Physics, Graduate School of Science, The University of Tokyo, Bunkyo-ku, Tokyo 113-0033, Japan}

\author{Ludovico Lami}
\email{ludovico.lami@gmail.com}
\affiliation{Institut f\"{u}r Theoretische Physik und IQST, Universit\"{a}t Ulm, Albert-Einstein-Allee 11, D-89069 Ulm, Germany}
\affiliation{QuSoft, Science Park 123, 1098 XG Amsterdam, the Netherlands}
\affiliation{Korteweg-de Vries Institute for Mathematics, University of Amsterdam, Science Park 105-107, 1098 XG Amsterdam, the Netherlands}
\affiliation{Institute for Theoretical Physics, University of Amsterdam, Science Park 904, 1098 XH Amsterdam, the Netherlands}


\begin{abstract}
We propose an approach to the study of quantum resource manipulation based on the basic observation that quantum channels which preserve certain sets of states are contractive with respect to the base norms induced by those sets. We forgo the usual physical assumptions on quantum dynamics: instead of enforcing complete positivity, trace preservation, or resource-theoretic considerations, we study transformation protocols as norm-contractive maps. This allows us to apply to this problem a technical toolset from functional and convex analysis, unifying previous approaches and introducing new families of bounds for the distillable resources and the resource cost, both one-shot and asymptotic. Since our expressions lend themselves naturally to single-letter forms, they can often be calculated in practice; by doing so, we demonstrate with examples that they can yield the best known bounds on quantities such as the entanglement cost. As applications, we not only give an alternative derivation of the recent result of [\href{https://www.nature.com/articles/s41567-022-01873-9}{Nat.\ Phys.\ 19, 184--189 (2023)}] which showed that entanglement theory is asymptotically irreversible, but also provide the quantities introduced in that work with explicit operational meaning in the context of entanglement distillation through a variation of the hypothesis testing relative entropy. Besides entanglement, we reveal a new irreversible quantum resource: through improved bounds for state transformations in the resource theory of magic-state quantum computation, we show that there exist qutrit magic states that cannot be reversibly interconverted under stabiliser protocols.
\end{abstract}

\maketitle

\tableofcontents

\section{Introduction}

The practical use of various physical features of quantum mechanical systems to enhance tasks such as computation, communication, or information processing is underlain by our ability to transform these resources efficiently. To understand the limits of such transformations has therefore been one of the most pressing problems of quantum information science. The ultimate extent of resource conversion is characterised by asymptotic transformation rates --- that is, the question of how well we can transform quantum systems when allowed to use an unbounded number of their copies. The study of these rates forms the foundation of quantum information theory, embodied by problems such as the calculation of quantum channel capacities~\cite{lloyd_1997, shor_2002, devetak_2005-1} or the distillation~\cite{bennett_1996-1} and dilution~\cite{bennett_1996} of quantum entanglement. Unfortunately, an exact understanding of the best achievable rates is often intractable, and the available bounds are, in many cases, quite loose. The problem is exacerbated by the wide variety of different settings and quantum resources beyond entanglement, each seemingly requiring specialised and unique approaches to fully characterise.

The unified framework of \emph{quantum resource theories}~\cite{horodecki_2012,chitambar_2019} was conceived to understand the common features of the various resources encountered in the processing of quantum systems, and has already provided significant insights into both asymptotic~\cite{brandao_2010-1, horodecki_2012, brandao_2015, kuroiwa_2020, ferrari_2020, regula_2021-1} and single-shot~\cite{gour_2017, takagi_2019, liu_2019, regula_2020, fang_2020, regula_2021-1} state transformations. Its remarkable generality has allowed for the consolidated description of many physical phenomena in a broad, axiomatic manner under a comprehensive mathematical formalism~\cite{coecke_2016, delrio_2015, gonda_2019}. 
However, this same generality can also be a drawback: in order to remain general, one can only make broadly applicable assumptions on the considered resources, which means that specific features of different resources might remain unaccounted for. This could suggest that it is not in our best interest to strive for such general approaches, and instead more tailored methods should be employed to obtain better results. We hereby put this premise into question by introducing an axiomatic framework that is strictly more general than previous approaches, yet it yields an improvement over known bounds for asymptotic resource transformations and reveals features of quantum resource manipulation that prior methods were unable to 
uncover.

\subsection{Quantum resources and their manipulation}\label{sec:intro_qres}

A quantum resource theory is a mathematical model of operational constraints encountered in manipulating quantum systems, stemming from either practical or fundamental limitations. It is specified by a set $\F$ of \emph{free states}, comprising certain density operators on a family of quantum systems, and a set $\O$ of \emph{free operations}, i.e.\ state-to-state transformations that connect those systems. Operationally, $\F$ and $\O$ represent states and operations that are relatively inexpensive to construct or implement in the laboratory; together, they form the basic ingredients of a quantum resource theory. In order to be physically realisable, any transformation in $\O$ should be a quantum channel, i.e.\ a completely positive and trace preserving linear map. In the resource theory of quantum entanglement~\cite{horodecki_2009}, for example, $\F$ is taken to be the set of separable states over bipartite quantum systems~\cite{werner_1989}, while historically $\O$ is usually taken to be the set of transformation implementable via local operations assisted by classical communication (LOCC)~\cite{chitambar_2014}, although other choices of operations have also found fruitful applications~\cite{rains_1999-1,audenaert_2003,brandao_2010,wang_2016,lami_2021-1}.

One of the central problems of quantum resource theories is to characterise the ultimate asymptotic transformation rates that are achievable by means of the free operations. In particular, one is interested in two opposite processes: on the one hand, that of \emph{distillation}~\cite{bennett_1996-1}, which consists of taking many copies of a given state $\rho$ and transforming them, via free operations, into as many copies as possible of some standard unit of resource $\phi$ (typically a pure state); and on the other hand, that of \emph{dilution}~\cite{bennett_1996}, in which we take as few copies as possible of $\phi$ and attempt to transform them into copies of $\rho$. In both cases, we allow for an asymptotically small error in the transformation, as quantified by the trace norm. The largest ratio between output copies of $\phi$ and input copies of $\rho$ in distillation is called the \emph{distillable resource} of $\rho$ under operations in $\O$, and denoted by $R_{d,\O}(\rho)$. Conversely, the smallest ratio between input copies of $\phi$ and output copies of $\rho$ in dilution is called the \emph{resource cost} of $\rho$ under operations in $\O$, and denoted by $R_{c,\O}(\rho)$.

For 
well-behaved sets of operations $\OO$ that are not unphysically permissive, variants of the `no-free-lunch' inequality $R_{c,\O}(\rho) \geq R_{d,\O}(\rho)$ can be shown to hold for all $\rho$. Notably, the inequality can be strict, and this fact is at the heart of the phenomenon of resource irreversibility~\cite{vidal_2001, yang_2005, wang_2017-1, lami_2021-1, lami_2022}. In the context of entanglement theory, when $\phi$ is chosen to be the `entanglement bit' (or \emph{ebit}), the above two measures turn into the distillable entanglement and the entanglement cost~\cite{horodecki_2009}; irreversibility in the context of entanglement manipulation under LOCC is a long-known fact~\cite{vidal_2001, yang_2005}, and even stronger variants of it were shown recently~\cite{wang_2017-1,lami_2021-1}.

More generally, computing either $R_{c,\O}(\rho)$ or $R_{d,\O}(\rho)$ for a given state $\rho$ is typically a challenging task. While upper bounds on $R_{c,\O}(\rho)$ or lower bounds on $R_{d,\O}(\rho)$ are relatively easy to obtain by the explicit construction of resource manipulation protocols, lower bounding $R_{c,\O}(\rho)$ or upper bounding $R_{d,\O}(\rho)$ requires a control over \emph{all} possible protocols --- in general, a highly non-trivial task.

Since the physical meaning and origin of $\O$ typically depends heavily on the specific type of quantum resource under examination,  
an approach often taken in the operational characterisation of quantum resources is to employ the natural axiom sometimes referred to as the `golden rule' of quantum resource theories~\cite{chitambar_2019}: a free operation should never generate the given resource, in the sense that if $\sigma \in \F$, then $\Lambda(\sigma) \in \F$. We denote by $\Omax \subseteq \CPTP$ the maximal set of all channels that satisfy this condition. This assumption is intended to be a weak and non-restrictive condition that any physical class of free transformations ought to obey, allowing one to describe all such operations in a single framework.

Beyond mathematical convenience, is such an approach physically meaningful? Take, as an example, the resource theory of entanglement: the operations $\Omax$ --- known as non-entangling or separability-preserving maps~\cite{brandao_2008-1, brandao_2010} --- include all naturally `free' classes of maps such as local operations and classical communication as well as all separable operations, but also transformations such as the swap channel, which can be considered costly to realise in practice because they require quantum communication. Thus, in practical scenarios, one is typically concerned with some physically motivated set of operations of interest $\O \subseteq \Omax$. It then follows immediately that --- for any possible choice of $\O$ --- the rate at which a resource can be distilled can only increase under $\Omax$, i.e.\ $R_{d,\Omax}(\rho) \geq R_{d,\O}(\rho)$, while the cost can only decrease, i.e.\ $R_{c,\Omax}(\rho) \leq R_{c,\O}(\rho)$. Importantly, this immediately gives
\begin{equation}\begin{aligned}\label{eq:shortcut}
	R_{c,\O}(\rho) \geq R_{c,\Omax}(\rho) \geq R_{d,\Omax}(\rho) \geq R_{d,\O}(\rho)
\end{aligned}\end{equation}
as long as the operations $\Omax$ satisfy the property that the distillable resource does not exceed the resource cost. The latter is a natural condition that is generally satisfied in all settings of physical interest, up to a possible appearance of some multiplicative normalisation factors~\cite{kuroiwa_2020, takagi_2021}.
We thus see immediately that studying the relation between resource distillation and dilution under $\Omax$ can immediately shed light on the corresponding relations under any physical class of operations $\O$. Notably, this is in spite of the fact that the operations in $\Omax$ may or may not be physical themselves.

\subsection{Basic idea}\label{sec:intuition}

We notice that a curious property has emerged in Eq.~\eqref{eq:shortcut}: although the operations $\O_{\max}$ are a relaxation of the physically relevant class $\O$, distillation protocols under $\O_{\max}$ provide potentially \emph{better} bounds on $R_{c,\O}$ than distillation protocols under $\O$ would. What if we took this insight further?

In this work, we thus take the unorthodox approach of improving on 
asymptotic transformation bounds by \emph{relaxing} the constraints on quantum resource manipulation \emph{even more} --- concretely, by considering a strictly larger set of maps than $\Omax$. But if $\Omax$ was intended to be the largest possible class of quantum channels relevant in the given resource theory, then how could we hope to define even larger classes? Our main idea may seem quite radical: first, we give up 
the physicality of the allowed maps, replacing complete positivity and trace preservation with contractivity with respect to the trace norm, expressed in formula as $\norm{\Lambda(X)}{1}\leq \norm{X}{1}$ for all operators $X$; furthermore, we renounce also the golden rule of resource non-generation, replacing it with contractivity with respect to a suitably chosen 
norm --- or, more generally, a gauge function --- $\norm{\cdot}{\mu}$, 
imposing that $\norm{\Lambda(X)}{\mu}\leq \norm{X}{\mu}$ for all $X$. The resulting enlarged set of free operations will be denoted by $\O_\mu \supset \Omax \supseteq \O$.

Following the basic idea outlined in Sec.~\ref{sec:intro_qres}, we will then aim to find bounds between the resource cost and distillable resource of the form $R_{c,\O}(\rho) \geq R_{c,\O_\mu}(\rho) \geq R_{d,\O_\mu}(\rho) \geq R_{d,\O}(\rho)$. We reemphasise our main point that, regardless of the unphysicality of $\O_\mu$, the quantity $R_{d,\O_\mu}(\rho)$ can give a better lower bound on the physically relevant quantity $R_{c,\O}(\rho)$ than $R_{d,\O}(\rho)$ would. Coupled with the fact that any ansatz of a distillation protocol in $\O_\mu$ provides a lower bound for $R_{d,\O_\mu}$, this can serve as a new, potentially stronger class of computable lower bounds for resource cost. 
An analogous statement holds for distillation: $R_{c,\O_\mu}(\rho)$ leads to a better upper bound on $R_{d,\O}(\rho)$ than $R_{c,\O}(\rho)$, and it is in turn straightforward to estimate from above.

In short, the starting point of this work is the realisation that going beyond physically allowed quantum channels and into the realm of maps that are not completely positive can yield improved bounds on operational asymptotic resource quantifiers. 
This programme will prove successful in several ways: most notably, the asymptotic bounds we obtain can go beyond all previously known bounds, such as ones based on the quantum relative entropy~\cite{horodecki_2001,horodecki_2012}, revealing features and limitations of quantum resources that remained completely out of reach of  previous theoretical approaches. 
One of the reasons behind the success of our methods is that by forgoing difficult-to-handle constraints such as complete positivity and resource non-generation, and focusing instead on simpler gauge-function--based constraints such as contractivity, we can employ a wealth of functional and convex analytic tools to tackle the problem. A by-product of this approach is that our results will hold for virtually every convex quantum resource theory.

\subsection{Technical hurdles and main contributions \label{sec:main_results}}

The above idealised picture, although appealing, contains 
some potential pitfalls in passing to the asymptotic limit.

First, for larger sets of operations such as $\O_\mu$, it is typically not at all easy to prove the no-free-lunch inequality $R_{c,\O_\mu}(\rho) \geq R_{d,\O_\mu}(\rho)$ that underpins this whole approach. This difficulty is primarily due to a potential asymptotic inequivalence of the trace norm $\norm{\cdot}{1}$ and the resource gauge $\norm{\cdot}{\mu}$.
Second, even if we could circumvent the above obstacle, we 
would still be left with the problem of lower bounding the asymptotic distillation rate $R_{d,\O_\mu}$ with a quantity that can actually be evaluated in practice. 
However, because the involved norms are typically not multiplicative under tensor products, this may require the evaluation of a regularised multi-copy quantity, which are typically not computable.
And finally, we have to be careful not to take a bound that is too loose --- we want to benefit from working in the extended framework of `unphysical' operations $\O_\mu$, and not merely recover previously known restrictions.

Overcoming these difficulties can be regarded as our main technical contribution. We will achieve this by identifying a novel 
variant of hypothesis testing relative entropy, dubbed the \emph{emancipated hypothesis testing relative entropy} (see 
Section~\ref{subsec:oneshot_distillation}), which we show to play a central role in characterising distillation under operations in $\O_\mu$. Our key insight is that this quantity obeys a particular type of continuity expressed by the `$\ve$-$\delta$ lemma' below (Lemma~\ref{lem:epsdelta}). By using this observation, we are able to provide a 
variant of the no-free-lunch relation $R_{c,\O_\mu}(\rho) \geq R_{d,\O_\mu}(\rho)$ that allows for straightforward lower bounding of the involved rates through suitable choices of convex gauge functions.

The results we will obtain in this way 
compete with state-of-the-art bounds on general transformation rates, recovering and improving on many of the leading asymptotic and non-asymptotic bounds for the manipulation of important quantum resources. We first demonstrate the effectiveness of our techniques by applying them to quantum entanglement~\cite{horodecki_2009}, in which context we provide in particular a number of lower bounds on the entanglement cost of any state that are efficiently computable via semidefinite programs. Obtaining such bounds is a notoriously thorny problem, especially since Hastings's counterexample~\cite{hastings_2009} has disproved the additivity conjectures~\cite{shor_2005} that would have allowed for a significant simplification, leaving us with intractable regularised expressions for asymptotic quantities. We then apply our toolset to the theory of magic states~\cite{veitch_2014, howard_2017} (non-stabiliser quantum computation), where our results immediately yield a number of bounds for the rates of transformations both in the case of many-qudit as well as many-qubit systems. We show that our results can strictly improve on previously known bounds in this setting, including ones based on the regularised relative entropy of magic~\cite{veitch_2014} and one-shot variants thereof~\cite{wang_2020}.

Crucially, the bounds established in our work are strong enough to reveal the \emph{irreversibility} of quantum resources, that is, the inability to reversibly interconvert two quantum states under the constraints of a given resource theory. This phenomenon --- although well known in some restricted settings such as the theory of entanglement under local operations and classical communication --- has been famously difficult to characterise in general~\cite{OpenProblemArxiv, brandao_2010-1, brandao_2008-1}. A pertinent aspect of the problem is that, if a given resource can be shown to be irreversible under all choices of free operations, then this entails that there cannot exist a single, unique monotone  that completely determines asymptotic convertibility of quantum states, which would have mirrored the fundamental role of entropy in many axiomatic approaches to thermodynamics~\cite{giles_1964, lieb_1999, Horodecki2002, vedral_2002, brandao_2015}. 
Only recently was it shown that such a strong irreversibility materialises for entanglement theory~\cite{lami_2021-1}: there exist no free operations that can make entanglement reversible, and so this theory cannot be governed by a single entropic quantity.
Our results can be used not only to recover the main finding of~\cite{lami_2021-1} in a somewhat different manner, but also to endow the monotones introduced in~\cite{lami_2021-1} with a precise operational meaning as the distillable entanglement under linear maps that contract the entanglement negativity.

Beyond entanglement theory, our general approach uncovers a new irreversible resource: 
we show the theory of magic-state quantum computation to be irreversible under all stabiliser operations, i.e.\ under all operations that can be implemented by using only Clifford unitaries, Pauli measurements, and ancillary stabiliser states. 
We do so by exemplifying a pair of qutrit pure magic states that cannot be converted reversibly using general types of free operations that include all stabiliser protocols. 
These results showcase the breadth and generality of our approach, obtained without detriment of its effectiveness --- indeed, the strength of our bounds compared to all previously known approaches allowed us to substantially advance our understanding of the problem of resource irreversibility, which no previous restrictions were able to do. 

The structure and main findings of our paper are as follows.

\begin{itemize}
    \item In Section~\ref{sec:prelim}, we set up the setting of our work and introduce all of the relevant concepts in the description of quantum resource theories and their manipulation. 
    \item In Section~\ref{sec:oneshot}, we introduce tight bounds on resource conversion in the one-shot setting, characterising exactly the restrictions on achievable regimes in resource transformations (Theorems~\ref{thm:oneshot_dilution} and~\ref{thm:oneshot_distillation}). We broadly divide our results between bounds that apply to the task of distillation (purification) as well as more general bounds, applicable in particular to the task of dilution (reverse of distillation). Our description of distillation protocols in Theorem~\ref{thm:oneshot_distillation} is enabled by a 
    generalisation of the hypothesis testing relative entropy that we introduce in Section~\ref{subsec:oneshot_distillation}. The key insights here are precise trade-offs between the performance of resource manipulation protocols and the values of related gauge-function--based resource quantifiers (Lemma~\ref{lem:epsdelta}), which will later form the foundation of our asymptotic bounds.
    \item Section~\ref{sec:asymptotic} is concerned with the establishment of our main results: a variety of bounds on the asymptotic rates of transformations between quantum states. First, in Section~\ref{subsec:regularised} we introduce our strongest bounds that have potential to improve on previously known limitations for resource conversion (Theorems~\ref{thm:cost_lowerbound} and~\ref{thm:dist_upperbound}). These bounds require the evaluation of asymptotic regularised quantities, which may be very difficult to compute in full generality. In Section~\ref{subsec:singleletter} we therefore propose a method to bound them by single-letter quantities, which yields a number of efficiently computable bounds on transformation rates that can be evaluated as convex or even semidefinite programs (notably, Corollary~\ref{cor:single_letter_gamma_bound}).  Section~\ref{subsec:mixtures} is then concerned with showing that one can, nevertheless, estimate the regularised quantities of Section~\ref{subsec:regularised} in some cases, allowing for the evaluation of our bounds when the single-letter approaches are not good enough. Finally, Section~\ref{subsec:positivity} discusses how appropriate choices of gauges in our expressions can be used to recover the best known efficiently computable bounds on resource distillation in the literature.
    \item Aiming to give more practical insight into the abstract resource-theoretic methods of the previous sections, Section~\ref{sec:apps} provides general pointers as to how our results can be applied to specific theories in practice and what such applications require.
    \item In Section~\ref{sec:app_ent}, we explicitly apply our results to the theory of entanglement manipulation. We discuss the relation between our bounds and known entanglement measures. We provide a class of new lower bounds on entanglement cost based on the reshuffling criterion for separability (Corollary~\ref{cor:reshuffled_lower_bound}). We finally show how the irreversibility of entanglement manipulation can be 
    retrieved with a somewhat different technique using our new results (Proposition~\ref{prop:entanglement_irrev}).
    \item Finally, Section~\ref{sec:app_magic} is concerned with the theory of magic-state quantum computation. We show how to apply our results in detail and provide several new bounds on the asymptotic transformation rates. 
    We establish, in particular, the irreversibility of the theory of magic for many-qudit systems (Theorem~\ref{thm:magic_irreversibility}), and discuss a conjecture related to an analogous irreversibility also in the case of many-qubit magic.
\end{itemize}

\section{Preliminaries}\label{sec:prelim}

Our considerations will take place in the space of self-adjoint operators acting on a separable Hilbert space. Specifically, let $\H$ be a separable Hilbert space, finite or infinite dimensional. We use $\T(\H)$ to denote the Banach space of all trace-class operators on $\H$, and $\B(\H)$ the dual space of all bounded operators on $\H$. For any $Y \in \B(\H)$ and $X \in \T(\H)$, we use the notation $\< Y, X \> \coloneqq \Tr Y^\dagger X$, which corresponds to the Hilbert--Schmidt inner product when $Y$ is also a Hilbert--Schmidt operator. The restrictions of $\T(\H)$ and $\B(\H)$ to self-adjoint operators are denoted as $\T_\sa(\H)$ and $\B_\sa(\H)$, respectively, but note that --- unless otherwise specified --- all of the operators we consider in this work will be self-adjoint, and we will sometimes not remark this explicitly. The sets $\T_+(\H)$ and $\B_+(\H)$ are the cones of positive self-adjoint operators in the respective space. Finally, $\D(\H)$ 
stands for the set of all density operators, that is, elements of $\T_+(\H)$ normalised to have unit trace. For any $\ket\psi \in \H$, we sometimes use the shorthand $\psi \coloneqq \proj{\psi} \in \T_+(\H)$.

Given two Hilbert spaces $\H$, $\H'$, we use $\CP(\H \to \H')$ to denote the set of all completely positive linear maps from $\T_\sa(\H)$ to $\T_\sa(\H')$; for simplicity, we will often omit the spaces and simply write $\CP$. The set of quantum channels (completely positive and trace preserving maps) will be denoted by $\CPTP$. All logarithms are to base 2.

\subsection{Asymptotic and non-asymptotic resource conversion}

The general approach to quantum resource theories begins with the identification of a particular set $\FF \subseteq \D(\H)$ of so-called \deff{free states}. They are typically understood as the states that are easy to prepare or available `for free' within the physical constraints of the given resource theory, without incurring additional costs. 
We will in particular be concerned with transformations of quantum states acting one space $\H_{\rm in}$ into states acting on a possibly different space $\H_{\rm out}$. 
The second ingredient is then the identification of a set of \deff{free operations} $\OO$, which are the transformations allowed in the given restricted setting. 
Since our approach will require flexibility in this choice, for now we make no assumptions about the structure or properties of such maps, only assuming them to be some subset of linear maps. 
Formally, a definition of a resource theory can be made as follows.
\begin{definition}
Given two Hilbert spaces $\H_{\rm in}$ and $\H_{\rm out}$, a \textbf{resource theory} is 
a triple $(\FF,\FF',\OO)$, where:
\begin{enumerate}[(i)]
\item 
$\FF = (\FF_n)_n$ and $\FF' = (\FF'_m)_m$ are sequences of sets of free states with $\FF_n \subseteq \D(\H_{\rm in}^{\otimes n})$ and $\FF'_m \subseteq \D(\H_{\rm out}^{\otimes m})$;
\item $\OO = (\OO_{(n,m)})_{n,m}$ is a double sequence of sets of free operations with $\OO_{(n,m)} \subseteq \mathcal{L}\big(\T_{\rm sa}(\H_{\rm in}^{\otimes n}) \to \T_{\rm sa}(\H_{\rm out}^{\otimes m})\big)$.
\end{enumerate}
With a slight abuse of notation, when no ambiguity arise we often use $\FF$ and $\OO$ to refer to any member of the corresponding family. For example, instead of $\rho^{\otimes n} \in \FF_n$ we will occasionally just write $\rho^{\otimes n} \in \FF$.
\end{definition}

The problem that we will focus on in particular concerns the asymptotic conversion of quantum systems. Let us define the rate of transforming a state $\rho\in \D\big(\H_{\mathrm{in}}\big)$ into another state $\sigma\in \D\big(\H_{\mathrm{out}}\big)$ 
by means of 
the designated class of free operations $\O$ as
\begin{equation}\begin{aligned} \label{eq:rate}
 r (\rho \toO \sigma) \coloneqq \sup \lsetr \liminf_{n\to\infty} \frac{m_n}{n} \barr 
 \lim_{n \to \infty} \inf_{\Lambda_n \in \O} \norm{\Lambda_n\left(\rho^{\otimes n}\right) - \sigma^{\otimes m_n}}{1} = 0 \rsetr,
\end{aligned}\end{equation}
where the supremum is over all sequences $(m_n)_n$ of positive integers. Here, each operation $\Lambda_n$ is a map $\T_\sa\!\left(\H_{\mathrm{in}}^{\otimes n}\right) \to \T_\sa\!\left(\H_{\mathrm{out}}^{\otimes m_n}\right)$.
As no ambiguity arises here, we simply use the notation $\OO$ without an explicit reference to the underlying spaces, again implicitly assuming that each space $\T_\sa(\H^{\otimes n})$ has its own corresponding set of free states $\FF$ in the given resource theory.

In practical manipulation of quantum resources, one often identifies a \deff{reference state} $\phi$ --- often assumed to be a highly resourceful state, or some state that allows for the efficient use of a given resource --- and treats it 
as a standard `unit' of the given resource.
The most important tasks to study are then resource distillation (asymptotic conversion into $\phi$) and resource dilution (asymptotic conversion from $\phi$). Having fixed a choice of $\phi$, we define the \deff{distillable resource} $R_{d,\O}$ and the \deff{resource cost} $R_{c,\O}$ as
\begin{equation}\begin{aligned}\label{eq:dist_cost_def_rates}
	R_{d,\O} (\rho) \coloneqq r(\rho \toO \phi), \qquad R_{c,\O} (\rho) \coloneqq r(\phi \toO \rho)^{-1}.
\end{aligned}\end{equation}

To understand such conversion more precisely, we can first consider non-asymptotic transformations, where the number of copies of an input state is taken to be finite and we fix some allowed error in the conversion. We thus define the \deff{one-shot $\boldsymbol{\ve}$-error distillable resource} $R^{(1),\ve}_{d,\O}$ and \deff{one-shot $\boldsymbol{\ve}$-error resource cost} $R^{(1),\ve}_{c,\O}$ as
\begin{equation}\begin{aligned}
	R^{(1),\ve}_{d,\O} (\rho) &\coloneqq \sup \lsetr m \in \NN \barr \inf_{\Lambda \in \O} \,\frac12 \norm{\Lambda(\rho) - \phi^{\otimes m}}{1} \leq \ve \rsetr\\
	R^{(1),\ve}_{c,\O} (\rho) &\coloneqq \inf \lsetr m \in \NN \barr \inf_{\Lambda \in \O} \,\frac12 \norm{\Lambda\left(\phi^{\otimes m}\right) - \rho}{1} \leq \ve \rsetr.
\end{aligned}\end{equation}
We then have
\begin{equation}\begin{aligned}
	R_{d,\O} (\rho) &= \lim_{\ve \to 0}\, \liminf_{n \to \infty}\, \frac{1}{n} R^{(1),\ve}_{d,\O} \left(\rho^{\otimes n}\right),\\
	R_{c,\O} (\rho) &= \lim_{\ve \to 0}\, \limsup_{n \to \infty}\, \frac{1}{n} R^{(1),\ve}_{c,\O} \left(\rho^{\otimes n}\right).
\end{aligned}\end{equation}
The operation that involves the limit in $n$ in the above expressions is called a \deff{regularisation}. From the technical standpoint, it constitutes the main hurdle to the calculation of the distillable resource and of the resource cost. Since the disproof of the additivity conjectures~\cite{shor_2005} by Hastings~\cite{hastings_2009}, we know that in general it cannot be omitted.

We will also consider zero-error resource manipulation, where no error whatsoever is allowed in the transformations, even at the one-shot level; specifically,
\begin{equation}\begin{aligned}
	R^{(1), \mathrm{exact}}_{d,\O} (\rho) &\coloneqq \sup \lsetr m \in \NN \barr \inf_{\Lambda \in \O} \norm{\Lambda(\rho) - \phi^{\otimes m}}{1} = 0 \rsetr,\\
	R^{(1), \mathrm{exact}}_{c,\O} (\rho) &\coloneqq \inf \lsetr m \in \NN \barr \inf_{\Lambda \in \O} \norm{\Lambda\left(\phi^{\otimes m}\right) - \rho}{1} = 0 \rsetr,
\end{aligned}\end{equation}
and
\begin{equation}\begin{aligned}
	R^{\mathrm{exact}}_{d,\O} (\rho) &= \liminf_{n \to \infty} \,\frac{1}{n} R^{(1),\mathrm{exact}}_{d,\O} \left(\rho^{\otimes n}\right),\\
	R^{\mathrm{exact}}_{c,\O} (\rho) &= \limsup_{n \to \infty} \,\frac{1}{n} R^{(1),\mathrm{exact}}_{c,\O} \left(\rho^{\otimes n}\right).
\end{aligned}\end{equation}

\subsection{Gauge-based free operations}

\deff{Resource monotones} (or \deff{resource measures}) are functions $M: \D(\H) \to \RR_+ \cup \{+\infty\}$ that satisfy $M(\rho) \geq M(\Lambda(\rho))$ for any free operation $\Lambda \in \O$. This includes monotones based on the relative entropy and other distance measures, or functions based on various norms induced by the set $\FF$. 
Any such monotone identifies a class of operations which includes all maps $\OO$, but may be strictly larger: namely, the quantum channels contracting the given measure, that is,
\begin{equation}\begin{aligned}
	\O \subseteq \lset \Lambda : \T_\sa(\H_{\mathrm{in}}) \to \T_\sa(\H_{\mathrm{out}}) \bar \Lambda \in \CPTP,\; M(\Lambda(\rho)) \leq  M(\rho) \; \forall \rho \in \D(\H_{\mathrm{in}}) \rset.
\end{aligned}\end{equation}
We will in particular specialise to the case when $M(\rho) = \norm{\rho}{\mu}$ is given by the value of some norm on $\T_\sa(\H)$ or, more generally, a convex gauge function (Minkowski functional)~\cite{rockafellar_1970,bourbaki_2003}.  In order to ensure a proper normalisation of this quantity, we will also assume that its value is at least one for any quantum state. 

\begin{definition}\label{def:gauge}
Given a Hilbert space $\H$, a \deff{resource gauge} $\norm{\cdot}{\mu} : \T_{\sa}(\H) \to \RR_+ \cup \{+\infty\}$ is defined to be any absolutely homogeneous function (i.e., $\norm{t X}{\mu} = |t| \norm{X}{\mu}$) that satisfies the triangle inequality and furthermore obeys the normalisation constraint $\norm{\rho}{\mu} \geq 1$ for all states $\rho \in \D(\H)$.
\end{definition}
Here we use the evocative notation $\norm{\cdot}{\mu}$ because the considered gauges will typically be norms or at least seminorms in all finite-dimensional spaces --- only in the fully general case of infinite-dimensional theories will they diverge to infinity for some states~\cite{lami_2021}. We remark that any gauge function satisfying the constraints of Definition~\ref{def:gauge} is convex. 

Although we do not explicitly require this in the definition, the resource gauges studied in this work will be induced by the underlying constraints of some resource theory. An intuitive and often encountered example of a gauge-based resource monotone is the entanglement negativity~\cite{vidal_2002-1,plenio_2005}: for any bipartite quantum state, the value of $\norm{\rho^\Gamma}{1}$ where $\Gamma$ denotes partial transposition (see Section~\ref{sec:app_ent}) is a useful gauge of how entangled a given state is. Such gauge-based measures can be defined in any finite-dimensional convex resource theory~\cite{regula_2018}, as we will describe shortly, and the restriction to these functions will make it easier to treat also non-positive operators in the same formalism.

Let us follow the line of reasoning based on contractivity further. When acting on a general operator $X \in \T_\sa(\H)$, any map $\Lambda \in \CPTP$ acts as a contraction in trace norm: $\norm{\Lambda(X)}{1} \leq \norm{X}{1}$. Relaxing the assumption of complete positivity and assuming only this contractivity property, we end up with a general set of maps defined as follows.

\begin{definition}
Given two Hilbert spaces $\H_{\rm in}$ and $\H_{\rm out}$ and 
the sequences of resource gauges $\norm{\cdot}{\mu}$
defined in each space $ \T_\sa(\H_{\rm in}^{\otimes n})$ and $\T_\sa(\H_{\rm out}^{\otimes m})$, 
the set of \textbf{gauge-based free operations} $\OO_\mu$ is defined by
\begin{equation}\begin{aligned}
	\O_\mu &\coloneqq \lset \Lambda : \T_\sa(\H_{\mathrm{in}}) \to \T_\sa(\H_{\mathrm{out}}) \bar \norm{\Lambda(X)}{1} \leq \norm{X}{1} ,\; \norm{\Lambda(X)}{\mu} \leq \norm{X}{\mu} \; \forall X \; \T_\sa(\H_{\mathrm{in}}) \rset.
\end{aligned}\end{equation}
\end{definition}

When $\norm{\cdot}{\mu}$ is chosen to be a gauge which contracts under all free operations $\Omax$, we necessarily have that $\O_\mu \supset \Omax$ as desired.
Different choices of the gauge function $\norm{\cdot}{\mu}$ will lead to operations with different properties, with the basic underlying idea being the contractivity of the chosen norm. 

For any choice of $\norm{\cdot}{\mu}$, we define the function $\norm{\cdot}{\mu}^\circ : \B_\sa(\H) \to \RR_+ \cup \{+\infty\}$ as
\begin{equation}
    \norm{\cdot}{\mu}^\circ = \sup \lset \< \cdot, W \> \bar \norm{W}{\mu} \leq 1 \rset.
\end{equation}
This is the gauge dual to $\norm{\cdot}{\mu}$ or, more generally, the support function of the set $\lset W \bar \norm{W}{\mu} \leq 1 \rset$. The crucial property is the Cauchy--Schwarz inequality: $\< Q, X \> \leq \norm{Q}{\mu}^\circ \norm{W}{\mu}$ for any $Q \in \B_\sa(\H)$, $W \in \T_\sa(\H)$, which follows directly from the definition.

\section{One-shot resource manipulation}\label{sec:oneshot}

We first derive general conditions on the possibility of transforming a given state into, or from, a given reference state in the one-shot setting.

Given a state $\rho$, we define its $\ve$-ball in trace norm as $B_\ve(\rho) \coloneqq \lset X \bar \norm{X}{1} \leq 1, \; \norm{X - \rho}{1} \leq \ve \rset$. We re-emphasise that the operators $X$ allowed here are not assumed to be positive, but we always take them to be self-adjoint. Note that the trace distance in quantum information theory is typically defined as $\frac12 \norm{X-\rho}{1}$; 
such a definition is motivated by the Helstrom--Holevo theorem~\cite{HELSTROM,Holevo1976}, which however applies only when both $X$ and $\rho$ have unit trace. Since $X$ is not a state, we employ here the more relevant distance $\norm{X-\rho}{1}$, although we do explicitly account for the factor of $\frac12$ in the definitions of all operational quantities (cf.\ Section~\ref{sec:prelim}) for consistency with other results.

\subsection{One-shot dilution} \label{subsec:oneshot_dilution}

\begin{boxed}{white}
\begin{theorem}\label{thm:oneshot_dilution}
Let $\phi \in \D(\H_{\rm in})$ and $\rho \in \D(\H_{\rm out})$ be two quantum states. If there exists a map $\Lambda \in \O_\mu$ such that $\norm{\Lambda(\phi) - \rho}{1} \leq \ve$, then
\begin{equation}\begin{aligned}\label{eq:dilution_thm_eq1}
	\norm{\phi}{\mu} \geq \inf_{X \in B_\ve(\rho)} \norm{X}{\mu}.
\end{aligned}\end{equation}
Conversely, if
\begin{equation}\begin{aligned}\label{eq:dilution_thm_eq2}
	\frac{1}{\norm{\phi}{\mu}^{\circ}} > \inf_{X \in B_\ve(\rho)} \norm{X}{\mu},
\end{aligned}\end{equation}
then there exists a map $\Lambda \in \O_\mu$ such that $\norm{\Lambda(\phi) - \rho}{1} \leq \ve$. 

If the gauge $\norm{\cdot}{\mu}$ is lower semicontinuous in the weak* topology, then the infimum in~\eqref{eq:dilution_thm_eq2} is achieved, and the strict inequality can be replaced with a non-strict one. Precisely, in this case a sufficient condition for the existence of a map $\Lambda \in \O_\mu$ with $\norm{\Lambda(\phi) - \rho}{1} \leq \ve$ is
\begin{equation}\begin{aligned}\label{eq:dilution_thm_semicont}
	\frac{1}{\norm{\phi}{\mu}^{\circ}} \geq \min_{X \in B_\ve(\rho)} \norm{X}{\mu}.
\end{aligned}\end{equation}
\end{theorem}
\end{boxed}

The weak* topology in the above statement is induced on the Banach space of trace class operators by thinking of it as the dual to the Banach space of compact operators, equipped with the operator norm. Equivalently, it can be defined as the initial topology induced by the functions $\T(\H) \ni X\mapsto \Tr KX$, where $K$ is an arbitrary compact operator on $\H$.

Before proving the theorem, let us discuss its assumptions and immediate consequences. The inequality in~\eqref{eq:dilution_thm_eq1} will be of most interest to us, as it establishes a general bound that applies to \emph{any} manipulation protocol; this includes, in particular, any subset of the operations $\O_\mu$, and thus also physically relevant classes of quantum channels.

Although the theorem is stated for general input states, we have given the input state the suggestive name $\phi$ because the most important application of Theorem~\ref{thm:oneshot_dilution} for us will be in resource dilution --- where, we recall, the input is a reference pure state $\phi$. Comparing~\eqref{eq:dilution_thm_eq1} and~\eqref{eq:dilution_thm_eq2} we see that, in order to give necessary and sufficient conditions for resource dilution under $\O_\mu$, we need to understand precisely the relation between the gauge $\norm{\phi}{\mu}$ and its dual $1/\norm{\phi}{\mu}^{\circ}$. We will see that this relation forms the foundation of many of the results of this work. 

Crucially, the bounds of the theorem can indeed be tight. To see this, for the sake of this discussion let us assume the weak* lower semicontinuity of the gauge $\norm{\cdot}{\mu}$, which will indeed be satisfied in all cases of interest to us in this work. Now, if a given pure state satisfies $\norm{\phi}{\mu} = 1/\norm{\phi}{\mu}^{\circ}$, then the converse and achievability parts of Theorem~\ref{thm:oneshot_dilution} match, and we obtain a necessary and sufficient condition for approximate one-shot conversion from this state into any other state under the maps $\OO_\mu$. If we further assume that
\begin{equation}\begin{aligned}\label{eq:assumption_mult_strict}
	\norm{\phi^{\otimes m}}{\mu} = \frac{1}{\norm{\phi^{\otimes m}}{\mu}^{\circ}} = \left(\frac{1}{\norm{\phi}{\mu}^{\circ}}\right)^m = \norm{\phi}{\mu}^m,
\end{aligned}\end{equation}
then we can in fact give an exact expression for one-shot $\ve$-error resource cost of any state:
\begin{equation}\begin{aligned}
R^{(1),\ve}_{c,\O_\mu} (\rho) = \ceil{ \frac{ \inf_{X \in \B_{2\ve}(\rho)} \log \norm{X}{\mu} }{ \log \norm{\phi}{\mu} } }.
\end{aligned}\end{equation}
Is an assumption such as~\eqref{eq:assumption_mult_strict} physically justifiable? Indeed, there are important cases of resource theories in which it holds true, notably the theory where we choose $\norm{X}{\mu} = \norm{X^\Gamma}{1}$ (resource theory of entanglement with non-positive partial transpose) and $\phi$ is a maximally entangled pure state. Of particular importance is the fact that, for $\ve=0$, we get
\begin{equation}\begin{aligned}
	R^{(1),\mathrm{exact}}_{c,\O_\mu} (\rho) &= 
    \ceil{ \frac{ \log \norm{\rho}{\mu} }{ \log \norm{\phi}{\mu} }} ,
\end{aligned}\end{equation}
which gives an exact operational meaning to the quantity $\norm{\rho}{\mu}$ --- for non-positive partial transpose, this 
endows the well-known logarithmic negativity~\cite{vidal_2002,plenio_2005} with a new operational meaning as the one-shot zero-error entanglement cost under maps which contract the negativity.

More generally, we might not get a strict equality between the gauge $\norm{\cdot}{\mu}$ and its dual; nevertheless, as long as we can ensure that $\norm{\phi^{\otimes m}}{\mu}$ is of the same order in $m$ as $\frac{1}{\norm{\phi^{\otimes m}}{\mu}^{\circ}}$, i.e.
\begin{equation}\begin{aligned}
	\norm{\phi^{\otimes m}}{\mu} = \Theta\left(\frac{1}{\norm{\phi^{\otimes m}}{\mu}^{\circ}}\right)
\end{aligned}\end{equation}
in asymptotic `big O' notation, then this will suffice to allow us to exactly understand the \emph{asymptotic} properties of resource dilution under $\OO_\mu$. 
Generally speaking, this latter requirement is more realistic and easier to meet in applications --- it holds, for instance, in the resource theory of entanglement, where $\norm{X}{\mu}$ denotes the base norm with respect to separable states. 
We refer the reader to the forthcoming Section~\ref{sec:app_ent} for a detailed discussion on the applications to entanglement theory.

\begin{proof}[Proof of Theorem~\ref{thm:oneshot_dilution}]
As $\norm{\cdot}{\mu}$ is a monotone under $\O_\mu$ by definition, for any protocol $\Lambda \in \O_\mu$ such that $\Lambda(\phi) = X'$ with $X' \in B_\ve(\rho)$, we have
\begin{equation}\begin{aligned}
	\norm{\phi}{\mu} \geq \norm{X'}{\mu} \geq \inf_{X \in B_\ve(\rho)}\norm{X}{\mu}.
\end{aligned}\end{equation}
On the other hand, take any operator $X \in B_\ve(\rho)$ and define the map
\begin{equation}\begin{aligned}
	\Lambda(Z) = \< Z , \phi\>\, X.
\end{aligned}\end{equation}
For any $Z \in \T_{\sa}(\H)$, by the Cauchy--Schwarz inequality for $\norm{\cdot}{\mu}$ it holds that
\begin{equation}\begin{aligned}
	\norm{\Lambda(Z)}{\mu} &= \left|\< Z, \phi \>\right| \,\norm{X}{\mu} \leq \norm{Z}{\mu} \norm{\phi}{\mu}^\circ \norm{X}{\mu}
\end{aligned}\end{equation}
and similarly
\begin{equation}\begin{aligned}
	\norm{\Lambda(Z)}{1} &\leq \norm{Z}{1} \norm{\phi}{\infty} \norm{X}{1} \leq \norm{Z}{1}.
\end{aligned}\end{equation}
What this means is that, as long $\norm{X}{\mu} \leq 1/\norm{\phi}{\mu}^\circ$, the operation satisfies $\Lambda \in \O_\mu$ and performs the desired transformation $\Lambda(\phi) = X$. Optimising over $X \in B_\ve(\rho)$ shows the achievability of the stated bound.

Note that the trace norm ball $B_\ve(X)$ is weak* compact, because it is a rescaled and translated version of the unit ball of the trace norm, which is in turn weak* compact by the Banach--Alaoglu theorem. Hence, the minimisation on the right-hand side of~\eqref{eq:dilution_thm_semicont} is well defined, as it involves a lower semicontinuous function minimised over a weak*-compact set.
\end{proof}

\subsection{One-shot distillation} \label{subsec:oneshot_distillation}

In order to characterise one-shot resource distillation, we will use a quantity based on the hypothesis testing relative entropy $D^\ve_H$. Recall that, for two states $\rho$ and $\sigma$, their hypothesis testing relative entropy is given by~\cite{buscemi_2010, wang_2012}
\begin{equation}\begin{aligned}
	D_H^\ve(\rho \| \sigma) = - \log \min \lset \< Q, \sigma\> \bar 0 \leq Q \leq \id,\; \< Q, \rho\> \geq 1-\ve \rset.
\end{aligned}\end{equation}

As we will shortly see explicitly, the fact that the minimisation here is achieved follows from the weak*-compactness of the considered set, itself a consequence of the Banach--Alaoglu theorem (see the proof of Lemma~\ref{lem:etruscan_dual}). We now define a 
variant of this quantity, dubbed the \deff{emancipated hypothesis testing relative entropy}, as follows:\footnote{The symbol $\h$ is the letter 
`h' in the Etruscan alphabet. The Etruscans were an ancient and somewhat mysterious pre-Roman people inhabiting today’s Tuscany --- which is named after them. Their language is lost, except for the alphabet and a few words. Emperor Claudius, perhaps one of 
its last speakers, 
had written a treatise on Etruscan history, which has unfortunately not survived through the centuries. 
}
\begin{equation}\begin{aligned}
	D_\h^\ve(\rho \| X) \coloneqq - \log \min \lset \< Q, X\> \bar -\id \leq Q \leq \id,\; \<Q, \rho\> \geq 1-\ve \rset,
\label{eq:Etruscan_D}
\end{aligned}\end{equation}
where we take $\log(x) = -\infty$ when $x \leq 0$ for consistency. Here we do not stipulate any physical meaning of this quantity a priori --- it is difficult to interpret it as `hypothesis testing' since we are no longer assuming that $Q$ forms a part of a valid POVM, but it will nevertheless be crucial in the description of resource distillation.
For convenience, we also define the non-logarithmic variant
\begin{equation}\begin{aligned}
	d_\h^\ve(\rho \| X) &\coloneqq  2^{D_\h^\ve(\rho \| X)}\\
&= \left( \min \lset \< Q ,X \>_+ \bar -\id \leq Q \leq \id,\; \< Q, \rho \> \geq 1-\ve \rset \right)^{-1}
\end{aligned}\end{equation}
where $x_+ \coloneqq \max \{ x, 0 \}$, and $1/0 = \infty$. 

We note that $D_\h^\ve$ can alternatively be expressed as a function of the standard hypothesis testing relative entropy $D_H^\ve$: re-parametrising $Q\coloneqq 2Q'-\id$ with $Q'\in [0,\id]$ in the definition of $D_\h^\ve$, we have that
\begin{equation}\begin{aligned}
	d_\h^\ve(\rho \| X)^{-1} = 2 \, d_H^{\ve/2}(\rho \| X)^{-1} - \Tr X,
\end{aligned}\end{equation}
where $d_H^\ve(\rho \| X) \coloneqq 2^{D_H^\ve(\rho \| X)}$.

We now use this quantity to study the conditions for the existence of one-shot distillation protocols.

\begin{boxed}{white}
\begin{theorem}
Let $\rho \in \D(\H_{\rm in})$ be any state and $\psi \in \D(\H_{\rm out})$ be pure. If there exists a map $\Lambda \in \O_\mu$ such that $\norm{\Lambda(\rho)- \phi}{1} \leq \ve$, then
\begin{equation}\begin{aligned}
	\frac{1}{\norm{\phi}{\mu}^{\circ}} \leq \inf_{\norm{Z}{\mu} \leq 1} d_\h^\ve(\rho \| Z).
\end{aligned}\end{equation}
Conversely, if
\begin{equation}\begin{aligned}
	\norm{\phi}{\mu} \leq \inf_{\norm{Z}{\mu} \leq 1} d_\h^\ve(\rho \| Z),
\end{aligned}\end{equation}
then there exists a map $\Lambda \in \O_\mu$ such that $\norm{\Lambda(\rho)- \phi}{1} \leq \ve$.
\label{thm:oneshot_distillation}
\end{theorem}
\end{boxed}
As before, let us discuss what the result immediately tells us. First, under the assumption $\norm{\phi}{\mu} = 1/\norm{\phi}{\mu}^{\circ}$, we have a necessary and sufficient condition for approximate one-shot conversion from any state into the pure state $\phi$. If one additionally assumes that
\begin{equation}\begin{aligned}\label{eq:assumption_golden_state}
	\norm{\phi^{\otimes m}}{\mu} = \frac{1}{\norm{\phi^{\otimes m}}{\mu}^{\circ}} = \left(\frac{1}{\norm{\phi}{\mu}^{\circ}}\right)^m = \norm{\phi}{\mu}^m,
\end{aligned}\end{equation}
then Theorem~\ref{thm:oneshot_distillation} gives an exact expression for one-shot distillable resources:
\begin{equation}\begin{aligned}
 R_{d,\O}^{(1),\ve} (\rho) = \floor{ \frac{ \inf_{\norm{Z}{\mu} \leq 1} D_\h^{2\ve}(\rho \| Z) }{ \log \norm{\phi}{\mu} } }.
\end{aligned}\end{equation}
For the case of $\ve=0$, we obtain in particular that
\begin{equation}\begin{aligned}
	R_{d,\O}^{(1),\mathrm{exact}} (\rho) = \floor{ \frac{ \inf_{\norm{Z}{\mu} \leq 1} D_\h^{0}(\rho \| Z) }{ \log \norm{\phi}{\mu} } }.
\end{aligned}\end{equation}
As mentioned in Section~\ref{subsec:oneshot_dilution}, the assumption of~\eqref{eq:assumption_golden_state} might not be exactly satisfied in all theories of interest, although it is often asymptotically tight. In the study of partial transposition, where $\norm{X}{\mu} = \norm{X^\Gamma}{1}$ and~\eqref{eq:assumption_golden_state} is satisfied, we will 
explicitly show that the value of $\inf_{\norm{Z}{\mu} \leq 1} D_\h^{0}(\rho \| Z)$ corresponds to a quantity recently introduced in~\cite{lami_2021-1}: the \emph{tempered logarithmic negativity}. See the forthcoming Section~\ref{subsec:partial_transposition}.
Although originally introduced through a completely different reasoning, this quantity was found to be of fundamental importance in establishing the irreversibility of entanglement theory. Here, we have given this quantity an explicit operational meaning: it quantifies the zero-error distillable entanglement under the class of operations $\O_\mu$ which contract the negativity. Due to the importance of this quantity and its generalisations, we will return to it in Section~\ref{sec:asymptotic}.

In order to prove the theorem we will need the following lemma, which establishes a dual form of the optimisation problem that minimises $d_\h^\ve(\rho \| Z)$ over the unit ball of the chosen gauge function. This will form a key technical insight of many of our bounds.

\begin{lemma}\label{lem:etruscan_dual}
For any state $\rho \in \D(\H)$, any choice of a resource gauge $\norm{\cdot}{\mu}$, and any $\ve \in [0,1)$, it holds that
\begin{equation}\begin{aligned}\label{eq:dual_form_dH}
  \inf_{\norm{Z}{\mu} \leq 1} d_\h^\ve(\rho \| Z) = \max \lset \frac{1}{\norm{Q}{\mu}^\circ} \bar \norm{Q}{\infty}\leq 1,\; \<Q, \rho\> \geq 1-\ve \rset.
\end{aligned}\end{equation}
For $\ve=0$, this can also be expressed as
\begin{equation}\begin{aligned}\label{eq:etruscan_dual_zero}
		\inf_{\norm{Z}{\mu} \leq 1} d_\h^0(\rho \| Z) &= \max \lset \< W, \rho \> \bar \norm{W}{\mu}^\circ \leq 1,\; \norm{W}{\infty} \leq \< W, \rho \> \rset
\end{aligned}\end{equation}
\end{lemma}
\begin{proof}
Eq.~\eqref{eq:dual_form_dH} 
holds since
\begin{equation}\begin{aligned}
	\inf_{\norm{Z}{\mu} \leq 1} d_\h^\ve(\rho \| Z) &\texteq{(i)} \inf_{\norm{Z}{\mu} \leq 1} \left( \min \lset \< Q, Z\>_+ \bar \norm{Q}{\infty}\leq 1,\; \<Q, \rho\> \geq 1-\ve \rset \right)^{-1} \\
    &= \left( \sup_{\norm{Z}{\mu} \leq 1} \min \lset \< Q, Z\>_+ \bar \norm{Q}{\infty}\leq 1,\; \<Q, \rho\> \geq 1-\ve \rset \right)^{-1} \\
	&\texteq{(ii)} \left( \sup_{\norm{Z}{\mu} \leq 1} \min \lset \< Q, Z\> \bar \norm{Q}{\infty}\leq 1,\; \<Q, \rho\> \geq 1-\ve \rset \right)^{-1}\\
	&\texteq{(iii)}  \left( \min \lset \sup_{\norm{Z}{\mu} \leq 1} \< Q, Z\> \bar \norm{Q}{\infty}\leq 1,\; \<Q, \rho\> \geq 1-\ve \rset \right)^{-1}\\
	&= \left( \min \lset \norm{Q}{\mu}^\circ \bar \norm{Q}{\infty}\leq 1,\; \<Q, \rho\> \geq 1-\ve \rset \right)^{-1}\\
		&= \max \lset \frac{1}{\norm{Q}{\mu}^\circ} \bar \norm{Q}{\infty}\leq 1,\; \<Q, \rho\> \geq 1-\ve \rset.
\end{aligned}\end{equation}

Here, in~(i) we used compactness of the set $\lset Q\in \B_\sa(\H) \bar \norm{Q}{\infty}\leq 1, \; \< Q,\rho \> \geq 1-\ve \rset$ to ensure that the minimum is achieved, and in~(iii) we used Sion's minimax theorem~\cite{sion_1958}. Although straightforward in finite-dimensional spaces, these steps require some elaboration in the infinite-dimensional case. The set $B_{\infty} \coloneqq \lset Q \in \B_\sa(\H) = \T_\sa(\H)^* \bar \norm{Q}{\infty}\leq 1 \rset$ is compact in the weak* topology induced on $\B_\sa(\H)$ by its pre-dual $\T_\sa(\H)$ thanks to the Banach--Alaoglu theorem~\cite[Theorem~2.6.18]{MEGGINSON}; since for $\rho\in \T_\sa(\H)$ the function $\B_\sa(\H)\ni Q \mapsto \<Q,\rho\>$ is weak* continuous by definition, the set $\lset Q\in \B_\sa(\H) \bar \< Q,\rho \> \geq 1-\ve\rset$ is immediately seen to be weak* closed; being the intersection of a weak* compact and a weak* closed set, the set $\lset Q\in \B_\sa(\H) \bar \norm{Q}{\infty}\leq 1, \; \< Q,\rho \> \geq 1-\ve \rset$ is also weak* compact. Finally, $B_\mu \coloneqq \lset Z \in \T_\sa(\H) \bar \norm{Z}{\mu} \leq 1 \rset$ is convex. Together with the bilinearity and continuity of the map $(B_\mu, B_\infty) \ni (Z,Q) \mapsto \<Q,Z\>$ with respect to the product of the trace norm and the weak* topologies, a straightforward consequence of the definitions and of the fact that $B_\infty$ is bounded in operator norm, this ensures that the conditions of the minimax theorem are satisfied.
Finally, in~(ii) 
we used the fact that $\norm{0}{\mu} = 0$ by the homogeneity of $\norm{\cdot}{\mu}$ so that $Z = 0$ is always feasible and the optimal value of the sup-min problem will never be negative.

In the case $\ve=0$, we have that
\begin{equation}\begin{aligned}
		\inf_{\norm{Z}{\mu} \leq 1} d_\h^0(\rho \| Z) &= \max \lset \frac{1}{\norm{Q}{\mu}^\circ} \bar \norm{Q}{\infty} = 1,\; \<Q, \rho\> = 1 \rset\\
		&= \max \lset \< W, \rho \> \bar \norm{W}{\mu}^\circ = 1,\; \norm{W}{\infty} = \< W, \rho \> \rset\\
		&= \max \lset \< W, \rho \> \bar \norm{W}{\mu}^\circ \leq 1,\; \norm{W}{\infty} \leq \< W, \rho \> \rset
\end{aligned}\end{equation}
through a change of variables $Q = \norm{Q}{\mu}^\circ W$. Here in the last line we observed that $\<W, \rho\> \leq \norm{W}{\infty}$ for any $\rho$ by the Cauchy--Schwarz inequality, so it suffices to impose the opposite inequality.
\end{proof}

The statement of the theorem can then be shown as follows.

\begin{proof}[Proof of Theorem~\ref{thm:oneshot_distillation}]
On the one hand, the fact that any $\Lambda \in \O_\mu$ is a contraction with respect to the gauge $\norm{\cdot}{\mu}$ means that its adjoint map, $\Lambda^\dagger$, will necessarily contract the dual function $\norm{\cdot}{\mu}^\circ$. Specifically,
\begin{equation}\begin{aligned}
	\norm{\Lambda^\dagger(\phi)}{\mu}^\circ &= \sup_{\norm{Z}{\mu} \leq 1} \< \Lambda(Z), \phi \> \\
	&\leq \sup_{\norm{Z'}{\mu} \leq 1} \< Z', \phi \>\\
	&= \norm{\phi}{\mu}^\circ.
\end{aligned}\end{equation}
In an analogous way, we see that $\Lambda^\dagger$ is a contraction with respect to the operator norm $\norm{\cdot}{\infty}$.
Additionally, using the fact that $\phi$ is a pure state we have that
\begin{equation}\begin{aligned}
\norm{\phi - \Lambda(\rho)}{1} &= \sup \lset \< W, \phi - \Lambda(\rho) \> \bar \norm{W}{\infty} \leq 1 \rset\\
&\geq \< \phi, \phi - \Lambda(\rho) \>\\
&= 1 - \< \Lambda^\dagger(\phi), \rho \>.
\end{aligned}\end{equation}
This altogether means that any $\Lambda \in \O_\mu$ such that $\norm{\Lambda(\rho)-\phi}{1}\leq \ve$ gives a feasible solution to Eq.~\eqref{eq:dual_form_dH} as $Q = \Lambda^\dagger(\phi)$, yielding
\begin{equation}\begin{aligned}
	\inf_{\norm{Z}{\mu} \leq 1} d_\h^\ve(\rho \| Z) \geq \frac{1}{\norm{\phi}{\mu}^{\circ}}.
\end{aligned}\end{equation}
On the other hand, 
taking an optimal $Q$ in~\eqref{eq:dual_form_dH} with 
 $1/\norm{Q}{\mu}^\circ =  \inf_{\norm{Z}{\mu} \leq 1} d_\h^\ve(\rho \| Z)$, the map $\Lambda$ defined as
\begin{equation}\begin{aligned}\label{eq:map_dist}
	\Lambda(Z) = \< Q, Z \> \phi
\end{aligned}\end{equation}
can be seen to satisfy
\begin{align}
	\norm{\Lambda(Z)}{1} &\leq \norm{Z}{1} \norm{Q}{\infty} \norm{\phi}{1} \leq \norm{Z}{1} ,\\ 
	\norm{\Lambda(Z)}{\mu} &\leq \norm{Z}{\mu} \norm{Q}{\mu}^\circ \norm{\phi}{\mu} \leq \norm{Z}{\mu}
\end{align}
using that $1/\norm{Q}{\mu}^\circ \geq \norm{\phi}{\mu}$ by assumption. This implies that $\Lambda \in \O_\mu$, and verifying that any map $\Lambda$ in~\eqref{eq:map_dist} satisfies
\begin{equation}\begin{aligned}
\norm{\Lambda(\rho) - \phi}{1} = \left| \< Q, \rho \> - 1 \right| \leq \ve
\end{aligned}\end{equation}
because $\<Q,\rho\> \in [1-\ve,1]$ 
concludes the proof.

\end{proof}

\begin{remark}
The fact that
\begin{equation}\begin{aligned}\label{eq:tempered_def}
		\inf_{\norm{Z}{\mu} \leq 1} d_\h^0(\rho \| Z) = \max \lset \< W, \rho \> \bar \norm{W}{\mu}^\circ = 1,\; \norm{W}{\infty} = \< W, \rho \> \rset,
\end{aligned}\end{equation}
as shown in Lemma~\ref{lem:etruscan_dual}, makes explicit the equality between this problem and the class of tempered resource monotones introduced in~\cite{lami_2021-1}, which are defined exactly by the right-hand side of~\eqref{eq:tempered_def}.

We can also note that, due to the self-adjointness of $Q$, the constraints $\norm{Q}{\infty}\leq 1$ and $\<Q, \rho \> = 1$ mean that $\braket{\psi|Q|\psi}$ must uniformly equal 1 for any vector $\ket\psi \in \supp(\rho)$, wherefore the constraints are actually equivalent to imposing that $\id \geq Q \geq 2\Pi_\rho - \id$, 
where $\Pi_\rho$ is the projection onto the support of $\rho$. This means that we can equivalently write
\begin{equation}\begin{aligned}
	\inf_{\norm{Z}{\mu} \leq 1} d_\h^0(\rho \| Z) =  \max \lset \frac{1}{\norm{Q}{\mu}^\circ} \bar 2 \Pi_\rho - \id \leq Q \leq \id \rset,
\end{aligned}\end{equation}
which explicitly shows that computing $d_\h^0$ depends only on $\supp \rho$. Moreover, we notice that $\inf_{\norm{Z}{\mu} \leq 1} d_\h^0(\rho \| Z) = 1/\norm{\id}{\mu}^\circ$ for any $\rho$ of full support.
\end{remark}

\subsection{Relation between distillation and dilution: the \texorpdfstring{$\boldsymbol{\ve}$-$\boldsymbol{\delta}$}{ε-δ} lemma} \label{subsec:epsilon-delta}

Our approach to describing asymptotic resource manipulation will rely on a precise understanding of the relations and trade-offs between the achievable 
performances of one-shot distillation and dilution.

Previously, the question of one-shot manipulation received significant attention particularly in the resource theory of thermodynamics~\cite{horodecki_2013,yungerhalpern_2016}, where it was shown that the task of one-shot work extraction (distillation) is governed by the hypothesis testing relative entropy $D^\ve_H$, while the opposite task of one-shot work formation (dilution) is governed by the smooth max-relative entropy~\cite{datta_2009}
\begin{equation}\begin{aligned}
D^\ve_{\max} (\rho \| \sigma) \coloneqq \log \inf \lset \lambda \in \RR \bar \rho' \leq \lambda \sigma,\; \rho' \in B_{2\ve}(\rho) \cap \D(\H) \rset.
\end{aligned}\end{equation}
 The relations between these two fundamental quantities have been studied in great detail in a number of works~\cite{tomamichel_2013,dupuis_2012,tomamichel_2016,datta_2013-1,anshu_2019} (owing also to the applications of the two entropies in information theory and quantum cryptography~\cite{TOMAMICHEL}), which allowed for an exact understanding of the operational aspects of thermodynamics in both the one-shot and asymptotic setting.

Although the beginnings of the applications of $D^\ve_H$ and $D^\ve_{\max}$ to more general quantum resource theories can be traced back to the seminal works by Brand\~ao and Plenio~\cite{brandao_2010-1,brandao_2010}, the precise quantitative study of the non-asymptotic interrelations in the manipulation of quantum resources has only been undertaken very recently~\cite{liu_2019,regula_2020,wilde_2021,takagi_2021}. In these works, conceptually similar ideas based on trade-offs between $D^\ve_{\max}$ and $D^\ve_H$ were employed. However, it is already known that $D^\ve_{\max}$ is not the right quantity to consider in the study of resources such as quantum entanglement~\cite{brandao_2011,liu_2019}, where resource dilution is not governed by $D^\ve_{\max}$, but rather by a different quantity based on the standard robustness $R^s_\FF$.  
The results of previous works thus failed to establish the tightest possible bounds, in particular in the asymptotic setting, because they relaxed the problem to the study of $D^\ve_{\max}$, which can only provide a looser bound.

Here, we establish a bound which constrains the trade-off between smoothed gauges $\norm{\cdot}{\mu}$ and the function $D_\h^\delta$. We can regard the following result in at least two possible ways. On the one hand, it as a natural generalisation of the so-called `$\epsilon$-lemma' that played a key role in our proof of the fundamental irreversibility of entanglement theory~\cite[Lemma~S6]{lami_2021-1} (see also Section~\ref{subsec:partial_transposition} below). On the other, since we have seen that these two quantities govern one-shot resource dilution and distillation, respectively, we can consider them as more appropriate objects of study than the max-relative entropy and the hypothesis testing relative entropy in our framework. The result below can then be understood as an extension of the one shot yield--cost trade-off relations of Refs.~\cite{wilde_2021,takagi_2021} to our setting, where the free operations are contractions with respect to certain resource gauges.

\begin{boxed}{white}
\begin{lemma}[$\ve$-$\delta$-lemma]\label{lem:epsdelta}
Choose any $\ve, \delta \in [0,1)$ such that $\delta + \ve < 1$. Then
\begin{equation}\begin{aligned}
\inf_{X \in B_\ve(\rho)} \log \norm{X}{\mu} \geq \inf_{\norm{Z}{\mu} \leq 1} D_\h^\delta(\rho \| Z) + \log \left( 1 - \delta - \ve \right).
\end{aligned}\end{equation}
\end{lemma}
\end{boxed}
\begin{proof}
Recall from Lemma~\ref{lem:etruscan_dual} that
\begin{equation}\begin{aligned}
\inf_{\norm{Z}{\mu} \leq 1} D_\h^\delta(\rho \| Z) = \log \max \lset \frac{1}{\norm{Q}{\mu}^\circ} \bar \< Q, \rho \> \geq 1-\delta,\; \norm{Q}{\infty} \leq 1 \rset.
\end{aligned}\end{equation}
Take $Q$ as 
an operator optimal for the above and assume that $\norm{Q}{\mu}^\circ < \infty$ to avoid a trivial statement. Take now $X \in \T_{\sa}(\H)$ to be any operator such that $ \norm{X-\rho}{1} \leq \ve$. Then
\begin{equation}\begin{aligned}
	\norm{X}{\mu} &\geq
	\sup \lset \< X, W \> \bar \norm{W}{\mu}^\circ \leq 1 \rset\\
	&\geq \< X, \frac{Q}{\norm{Q}{\mu}^\circ} \> \\
	&= \< \rho, \frac{Q}{\norm{Q}{\mu}^\circ} \> - \< \rho - X, \frac{Q}{\norm{Q}{\mu}^\circ} \> \\
	&\geq (1-\delta) \frac{1}{\norm{Q}{\mu}^\circ}- \norm{\rho - X}{1} \frac{\norm{Q}{\infty}}{\norm{Q}{\mu}^\circ}\\
	&\geq (1-\delta - \ve) \frac{1}{\norm{Q}{\mu}^\circ},
\end{aligned}\end{equation}
where the first line is by weak Lagrange duality, and 
in the fourth line we used the Cauchy--Schwarz inequality (or, specifically, the H\"older inequality for $\norm{\cdot}{1}$). Optimising over all feasible $X$, 
we get the desired inequality. 
\end{proof}

\section{Asymptotic resource manipulation}\label{sec:asymptotic}

The purpose of this section is to employ the technical results of the above section to derive explicit bounds on the asymptotic rates of state transformations, focusing in particular on the distillation and dilution of a given resource. Our strategy follows the conceptual scheme summarised in Figure~\ref{fig:conceptual}: to obtain lower bounds on the efficiency of resource dilution, i.e.\ on the resource cost, we will connect it with the (smoothed) gauge $\norm{\cdot}{\mu}$ (Theorem~\ref{thm:oneshot_dilution}), then use the $\ve$-$\delta$-lemma (Lemma~\ref{lem:epsdelta}) to relate this with the emancipated hypothesis testing relative entropy $D^\delta_\h$~\eqref{eq:Etruscan_D}, and finally use this latter quantity to lower bound the cost. This is precisely the technique underlying the proof of Theorem~\ref{thm:cost_lowerbound} below. An analogous strategy is used to prove Theorem~\ref{thm:dist_upperbound}. The purpose of doing so, naturally, is that the smoothed emancipated hypothesis testing relative entropy, initially defined as an infimum, can alternatively be expressed as a supremum (see Lemma~\ref{lem:etruscan_dual}). By means of Theorem~\ref{thm:cost_lowerbound}, therefore, we can generate and explicitly compute lower bounds to the resource cost in a systematic fashion.

\begin{figure}[h]
\centering
\begin{tikzpicture}
\node (dist) at (0,0) {\large distillation};
\node (dil) at (6.472,0) {\large dilution};
\node (D) at (0,4) {\large $D^{\delta}_\h$};
\node (mu) at (6.472,4) {\large $\ \norm{\cdot}{\mu}$};
\draw[<->, thick] (D) -- node [above] {$\ve$-$\delta$-lemma (Lemma~\ref{lem:epsdelta})} (mu);
\draw[<->, thick] (D) -- node [anchor=center, left] {\rotatebox{90}{Theorem~\ref{thm:oneshot_distillation}}} (dist);
\draw[<->, thick] (mu) -- node [anchor=center, right] {\rotatebox{-90}{Theorem~\ref{thm:oneshot_dilution}}} (dil);
\draw[->, thick] (D) -- node [yshift=8pt] {\rotatebox{-32}{\phantom{lorem ipsum suas p} Theorem~\ref{thm:cost_lowerbound}}} (dil);
\draw[->, thick] (mu) -- node [yshift = 8pt] {\rotatebox{32}{Theorem~\ref{thm:dist_upperbound} \phantom{lorem ipsum suas }}} (dist);
\end{tikzpicture}
\caption{The quantities considered in this work, i.e.\ the emancipated hypothesis testing relative entropy $D^{\delta}_\h$~\eqref{eq:Etruscan_D} and the resource gauge $\norm{\cdot}{\mu}$, and the main conceptual relations between them.
}
\label{fig:conceptual}
\end{figure}
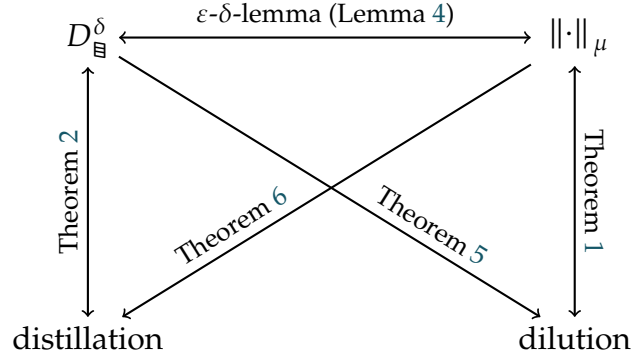

\subsection{Bounds on distillable resource and resource cost}\label{subsec:regularised}

Let us begin with an explicit investigation of how the one-shot results obtained in Section~\ref{sec:oneshot} can help us shed light on the asymptotic properties of resource manipulation. Although the bounds in this section are defined through asymptotic regularisations which do not necessarily have known computable forms, 
they provide important conceptual generalisations of other bounds that have appeared previously. We will shortly see how to use the 
results presented here to obtain also computable, single-letter bounds, as well as how the regularised quantities themselves can be bounded in some cases.

We start with a general estimate of the efficiency of resource dilution.

\begin{boxed}{white}
\begin{theorem}\label{thm:cost_lowerbound}
For any state $\rho$, any pure state $\phi$, and any $\delta \in [0,1)$, the rate of transformation from $\phi$ to $\rho$ under $\O_\mu$ satisfies
\begin{equation}\begin{aligned}
	r(\phi \toOmu \rho)^{-1} &\geq \frac{ \limsup_{n\to\infty} \frac1n \inf_{\norm{Z}{\mu}\leq 1} D^{\delta}_\h \left(\rho^{\otimes n} \| Z\right)}{ \loginf{\phi} }\\
	&= \frac{ \limsup_{n\to\infty} \frac{1}{n} \log \sup \lset \frac{1}{\norm{Q}{\mu}^\circ} \bar \norm{Q}{\infty}\leq 1,\; \<Q, \rho^{\otimes n}\> \geq 1-\delta \rset }{ \loginf{\phi} },
\label{eq:cost_lowerbound}
\end{aligned}\end{equation}
where
\begin{equation}\begin{aligned}
	\loginf{\phi} &\coloneqq \limsup_{n \to \infty} \frac{1}{n} \log \norm{\phi^{\otimes n}}{\mu}.
\end{aligned}\end{equation}
\end{theorem}
\end{boxed}

This bound has the potential to improve on known lower bounds on resource cost $R_{c,\O_\mu}(\rho) = r_{\O_\mu}(\phi \to \rho)^{-1}$, and in particular it majorises 
the class of bounds which was used recently in~\cite{lami_2021-1} to show the irreversibility of entanglement theory.
The quantity $\loginf{\phi}$ represents here a parameter that characterises the resource properties of the given reference state $\phi$, and we will see that, in relevant cases, it can be straightforwardly computed or bounded.

However, a major question is: can the regularisation of the emancipated hypothesis testing relative entropy $D^{\delta}_\h$ be evaluated in practice? Although we are unable to provide an answer to this problem at this stage, 
we do remark the conceptual similarity of this bound to the asymptotic bound of Brand\~ao and Plenio~\cite{brandao_2010}; indeed, the quantity encountered in their work,
\begin{equation}\begin{aligned}\label{eq:bp}
	\limsup_{n\to\infty} \frac1n \inf_{\sigma \in \FF} D^{\delta}_H (\rho^{\otimes n} \| \sigma),
\end{aligned}\end{equation}
can also be used as a lower bound for resource cost, but it can never be better than our result in Theorem~\ref{thm:cost_lowerbound}, since $D^{\delta}_\h (\rho^{\otimes n} \| \sigma) \geq D^{\delta}_H (\rho^{\otimes n} \| \sigma)$ by definition. Interestingly, the quantity in~\eqref{eq:bp} was conjectured~\cite{brandao_2010, brandao_2010-1, berta_gap} and finally very recently proven~\cite{hayashi_stein, blurring} to reduce (under mild assumptions on $\FF$) to a well-known resource monotone: the regularised relative entropy of the given resource. 
Our Theorem~\ref{thm:cost_lowerbound} then has the potential to outperform relative-entropy--based bounds in general. As we will shortly see, there are indeed examples in the resource theories of entanglement and magic where the emancipated hypothesis testing relative entropy yields a \emph{strict} improvement over all such previously known bounds.

It is therefore a very interesting open problem to determine whether an exact form of the asymptotic regularisation for the quantity $D^\delta_\h$ in~\eqref{eq:cost_lowerbound} can be established, and what the equivalent of relative entropy would be in this setting. This would lead to a new entropic-like lower bound on resource dilution cost that can be strictly better than the regularised relative entropy. 

\begin{proof}[Proof of Theorem~\ref{thm:cost_lowerbound}]
Assume that $C$ is any achievable rate of dilution, and consider a sequence of operations $(\Lambda_n)_n \in \O_\mu$ such that $\frac12 \norm{\Lambda(\phi^{\otimes \ceil{C n}}) - \rho^{\otimes n}}{1} \leq \ve_n$, with the error $\ve_n$ asymptotically vanishing. From Theorem~\ref{thm:oneshot_dilution} we have that
\begin{equation}\begin{aligned}
\norm{\phi^{\otimes \ceil{C n}}}{\mu} \geq \inf_{X \in \B_{\ve_n}(\rho^{\otimes n})} \norm{X}{\mu},
\end{aligned}\end{equation}
Using the $\ve$-$\delta$-lemma (Lemma~\ref{lem:epsdelta}), we get that
\begin{equation}\begin{aligned}\label{eq:secondlaw_upper}
	\norm{\phi^{\otimes \ceil{C n}}}{\mu} \geq \inf_{\norm{Z}{\mu}\leq 1} d_\h^{\delta}(\rho^{\otimes n} \| Z) \left( 1 - \delta - 2\ve_n \right)
\end{aligned}\end{equation}
for $n$ sufficiently large and $\delta$ sufficiently small so that $\delta + 2\ve_n < 1$. Thus
\begin{equation}\begin{aligned}
	  \frac{1}{n} \log \norm{\phi^{\otimes \ceil{C n}}}{\mu} &\geq \frac1n \inf_{\norm{Z}{\mu}\leq 1} D_\h^{\delta}(\rho^{\otimes n} \| Z) + \frac{\log\left( 1 - \delta - 2\ve_n \right)}{n}.
	  \end{aligned}\end{equation}
We can rewrite this as
\begin{equation}\begin{aligned}
	\frac{\ceil{Cn}}{n} \frac{1}{\ceil{C n}} \log \norm{\phi^{\otimes \ceil{C n}}}{\mu} &\geq \frac1n \inf_{\norm{Z}{\mu}\leq 1} D_\h^{\delta}(\rho^{\otimes n} \| Z) + \frac{\log\left( 1 - \delta - 2\ve_n \right)}{n}.
\end{aligned}\end{equation}
Taking $
\limsup_{n\to\infty}$ of both sides yields
\begin{equation}\begin{aligned}
	C \, \limsup_{n\to\infty} \frac{1}{\ceil{C n}} \log \norm{\phi^{\otimes \ceil{C n}}}{\mu} \geq 
 \limsup_{n\to\infty} \frac1n \inf_{\norm{Z}{\mu}\leq 1} D^{\delta}_\h \left(\rho^{\otimes n} \| Z\right).
\end{aligned}\end{equation}
Upper bounding $\displaystyle \limsup_{n\to\infty} \frac{1}{\ceil{C n}} \log \norm{\phi^{\otimes \ceil{C n}}}{\mu} \leq \loginf{\phi}$ concludes the proof of the first line of~\eqref{eq:cost_lowerbound}. The equality on the second line follows from the dual expression derived in Lemma~\ref{lem:etruscan_dual}.
\end{proof}

\begin{remark}
We note that our definition of the rate $r(\phi \toOmu \rho)$ assumes that the transformation error $\ve_n$ vanishes asymptotically. In the proof of the theorem, this is not strictly required: all we need here is for the error to eventually become sufficiently small so that $\delta + 2 \ve_n < 1$, allowing for the crucial $\ve$-$\delta$-lemma (Lemma~\ref{lem:epsdelta}) to be applied. If we instead 
took the limit as $\delta \to 0$ of the resulting expression (which exists because $D^\delta_\h$ is explicitly non-decreasing in $\delta$), this would then give us the bound
\begin{equation}\begin{aligned}\label{eq:delta_to_zero}
 r(\phi \toOmu \rho)^{-1} \geq \tilde{r}(\phi \toOmu \rho)^{-1} &\geq \frac{ \lim_{\delta \to 0}  \limsup_{n\to\infty} \frac1n \inf_{\norm{Z}{\mu}\leq 1} D^{\delta}_\h \left(\rho^{\otimes n} \| Z\right)}{ \loginf{\phi} },
\end{aligned}\end{equation}
which 
applies even to a modified notion of rate $\tilde{r}(\phi \toOmu \rho)$, defined analogously to~\eqref{eq:rate} but where the trace norm error $\ve_n = \frac12 \norm{\Lambda_n(\phi^{\otimes n}) - \rho^{\otimes \floor{rn}}}{1}$ satisfies only $\limsup_{n\to\infty} \ve_n < 1/2$ instead of $\lim_{n\to\infty} \ve_n = 0$ as in~\eqref{eq:rate}. 
Such a bound is known as a `pretty strong converse' bound~\cite{morgan_2014} (cf.~\cite{lami_2021-1} in the context of distillation and dilution of quantum resources). For simplicity, here we only explicitly consider rates with vanishing error.

We also note that taking the limit $\delta \to 0$ could be expected to make the regularisation in Eq.~\eqref{eq:delta_to_zero} easier to compute than the $\delta$-dependent form.
\end{remark}

A similar type of bound can be established for distillable resource yield.

\begin{boxed}{white}
\begin{theorem}\label{thm:dist_upperbound}
For any state $\rho$, any pure state $\phi$, and any $\ve \in [0,1)$,  the rate of transformation from $\rho$ to $\phi$ under $\O_\mu$ satisfies
\begin{equation}\begin{aligned} \label{eq:dist_upperbound}
	r(\rho \toOmu \phi) \leq \frac{\liminf_{n\to\infty} \frac1n \inf_{X \in \B_{\ve}(\rho^{\otimes n})} \log \norm{X}{\mu} }{ -\loginfd{\phi} } ,
\end{aligned}\end{equation}
where
\begin{equation}\begin{aligned}
	\loginfd{\phi} &\coloneqq \limsup_{n \to \infty} \frac{1}{n} \log \norm{\phi^{\otimes n}}{\mu}^\circ.
\end{aligned}\end{equation}
For consistency, if $\norm{\phi^{\otimes n}}{\mu}^\circ = 1$, we understand the bound as the trivial restriction $r(\rho \toOmu \phi) \leq \infty$.
\end{theorem}
\end{boxed}

\begin{proof}
Let $D$ be an achievable rate for distillation, and consider a sequence of protocols $(\Gamma_n)_n \in \O_\mu$ such that $\frac12 \norm{\Gamma_n(\rho^{\otimes n}) - \phi^{\otimes \floor{Dn}}}{1} \coloneqq \delta_n$, with $\delta_n$ denoting the asymptotically vanishing distillation error. Theorem~\ref{thm:oneshot_distillation} then gives
\begin{equation}\begin{aligned}\label{eq:secondlaw_lower}
	\frac{1}{\norm{\phi^{\otimes \floor{D n}}}{\mu}^{\circ}} \leq \inf_{\norm{Z}{\mu} \leq 1} d_\h^{2 \delta_n}(\rho^{\otimes n} \| Z).
\end{aligned}\end{equation}
For $n$ sufficiently large and $\ve$ sufficiently small so that $2 \delta_n + \ve < 1$, the $\ve$-$\delta$-lemma (Lemma~\ref{lem:epsdelta}) allows us to estimate 
\begin{equation}\begin{aligned}
	\frac{\floor{Dn}}{n} \frac{1}{\ceil{Dn}} \left(-\log \norm{\phi^{\otimes \ceil{D n}}}{\mu}^\circ\right) \leq \frac1n \inf_{X \in \B_{\ve_n}(\rho^{\otimes n})} \log \norm{X}{\mu} - \frac{\log(1- 2 \delta_n - \ve)}{n}.
\end{aligned}\end{equation}
Taking $\lim_{\ve\to 0}\liminf_{n\to\infty}$ of both sides and bounding
\begin{equation}\begin{aligned}
	\liminf_{n\to\infty}\frac{1}{\floor{Dn}} \left( - \log \norm{\phi^{\otimes \floor{D n}}}{\mu}^{\circ} \right) \geq \liminf_{n\to\infty}\frac{1}{n} \left( - \log \norm{\phi^{\otimes n}}{\mu}^{\circ} \right) = - \loginfd{\phi}
\end{aligned}\end{equation}
gives the stated result.
\end{proof}

\begin{remark}
Again, here we do not actually need the transformation error 
$\delta_n \coloneqq \frac12 \norm{\Gamma_n(\rho^{\otimes n}) - \phi^{\otimes \floor{Dn}}}{1}$ to go to 0 asymptotically; if we required that $\ve \to 0$ instead, this would yield the pretty strong converse bound
\begin{equation}\begin{aligned}
	r(\rho \toOmu \phi) \leq \tilde{r}(\rho \toOmu \phi) \leq \frac{\lim_{\ve \to 0} \liminf_{n\to\infty} \frac1n \inf_{X \in \B_{\ve}(\rho^{\otimes n})} \log \norm{X}{\mu} }{ -\loginfd{\phi} },
\end{aligned}\end{equation}
which would apply for a modified definition of rate $\vphantom{\Big|} \tilde{r}(\rho \toOmu \phi)$ where we only require the transformation error to satisfy $\limsup_{n\to\infty} \delta_n < 1/2$ instead of $\lim_{n\to\infty} \delta_n = 0$. 
As before, this could potentially constitute a more easily computable formula.
\end{remark}

\begin{remark}
One may wonder whether our bounds in Theorems~\ref{thm:cost_lowerbound} and \ref{thm:dist_upperbound} are simply a statement of the `no free lunch' relation $R_{c,\O_\mu}(\rho) \geq R_{d,\O_\mu}(\rho)$, that is, the fact that one cannot extract more resources from a given state than it costs to produce that same state. As expounded in Sec.~\ref{sec:intuition}, this is certainly the conceptual intuition behind our approach: we use the quantity $D^\ve_\h$, which characterises distillation (cf.\ Theorem~\ref{thm:oneshot_distillation}), to lower bound the resource cost; conversely, we use $\norm{\cdot}{\mu}$, which directly quantifies resource dilution (cf.\ Theorem~\ref{thm:oneshot_dilution}), to upper bound distillable resources. 
However, our method does not directly go through $R_{d,\OO_\mu}$ to lower bound $R_{c,\OO_\mu}$ or vice versa; through the application of the $\ve$-$\delta$ lemma (Lemma~\ref{lem:epsdelta}) we circumvent the need to rely on such a relation, and at the same time we also obtain more insight into the transformation errors of the protocols.

In fact, it is a noteworthy observation that the 
inequality $R_{c,\O_\mu}(\rho) \geq R_{d,\O_\mu}(\rho)$ cannot be proved directly with the 
same methods. As a matter of fact, we do not yet know whether such inequality holds in full generality. Let us examine 
more closely the reason why this is so. Given a state $\rho$, consider a pair of achievable rates $C,D$ for dilution and distillation, respectively. According to~\eqref{eq:rate}, this means that there exist operations $\Lambda_n\in \O_\mu$ and $\Gamma_n\in \O_\mu$ with the property that $\ve_n \coloneqq \norm{\Lambda_n(\phi^{\otimes n}) - \rho^{\otimes \floor{Cn}}}{1}$ and $\delta_n \coloneqq \norm{\Gamma_n(\rho^{\otimes n}) - \phi^{\otimes \floor{Dn}}}{1}$ satisfy $\lim_{n\to\infty} \ve_n = 0 = \lim_{n\to\infty} \delta_n$. A standard way of proceeding would be to concatenate these operations and conclude that if $CD>1$ then we would be able to distil more copies of $\phi$ than we started with, something that should be impossible with free operations. We would therefore conclude that $CD \leq 1$, and upon taking the supremum over achievable rate pairs this would yield that $R_{d,\OO_\mu} (\rho) \leq R_{c,\OO_\mu} (\rho)$, as desired.

Let us try to make this reasoning rigorous. Defining $\Theta_n \coloneqq \Gamma_{\floor{Cn}} \circ \Lambda_n \in \O_\mu$, we have immediately that
\begin{equation}\begin{aligned}
\norm{\Theta_n(\phi^{\otimes n}) - \phi^{\otimes \floor{D\floor{Cn}}}}{1} &\leq \norm{\Theta_n(\phi^{\otimes n}) - \Gamma_{\floor{Dn}}(\rho^{\otimes \floor{Cn}})}{1} + \norm{\Gamma_{\floor{Dn}}(\rho^{\otimes \floor{Cn}}) - \phi^{\otimes \floor{D\floor{Cn}}}}{1} \\
&= \norm{\Gamma_{\floor{Cn}}\left(\Lambda_n(\phi^{\otimes n}) - \rho^{\otimes \floor{Cn}}\right)}{1} + \norm{\Gamma_{\floor{Dn}}(\rho^{\otimes \floor{Cn}}) - \phi^{\otimes \floor{D\floor{rn}}}}{1} \\
&\leq \ve_n + \delta_{\floor{Cn}} ,
\end{aligned}\end{equation}
where in the last line we used contractivity of $\norm{\cdot}{1}$ under operations in $\O_\mu$. Leveraging also contractivity of $\norm{\cdot}{\mu}$, we obtain that $\norm{\Theta_n(\phi^{\otimes n})}{\mu} \leq \norm{\phi^{\otimes n}}{\mu}$. Therefore, $\phi^{\otimes \floor{D\floor{Cn}}}$ lies to within a vanishingly small trace distance from a state whose $\mu$-norm is upper bounded by $\norm{\phi^{\otimes n}}{\mu}$.

Importantly, however, this does not allow us to draw any conclusions concerning the $\mu$-norm of $\phi^{\otimes \floor{D\floor{Cn}}}$ itself. Indeed, to do so one could be tempted to apply the triangle inequality
\begin{equation}\begin{aligned}
\norm{\phi^{\otimes \floor{D\floor{Cn}}}}{\mu} \leq \norm{\phi^{\otimes n}}{\mu} + \norm{\Theta_n(\phi^{\otimes n}) - \phi^{\otimes \floor{D\floor{Cn}}}}{\mu};
\end{aligned}\end{equation}
however, in order to further upper bound $\norm{\Theta_n(\phi^{\otimes n}) - \phi^{\otimes \floor{D\floor{Cn}}}}{\mu}$, we would need to somehow turn the $\mu$-gauge into the trace norm, 
as we know the latter to be vanishingly small. 
This would seem in principle possible, for instance, if $\norm{\cdot}{\mu}$ was a bona fide norm and we worked in finite-dimensional spaces, where all norms are equivalent. Doing so, one would obtain
a relation of the form 
\begin{equation}\begin{aligned}
\norm{\Theta_n(\phi^{\otimes n}) - \phi^{\otimes \floor{D\floor{Cn}}}}{\mu}\leq N_n \norm{\Theta_n(\phi^{\otimes n}) - \phi^{\otimes \floor{D\floor{Cn}}}}{1} \leq N_n \left(\ve_n + \delta_{\floor{Cn}}\right)
\end{aligned}\end{equation}
for some constant $N_n>0$. We would be able to complete the proof if we could show that the right-hand side of the above inequality converges to zero. The key issue that prevents us from doing so is that the constant $N_n$ can diverge as $n\to\infty$ --- and, typically, it will do so exponentially fast.

The above reasoning should convince the reader that the na\"ive way of proving the no free lunch relation fails in this case. Although we do not expect this inequality to hold exactly in all resource theories,  
it turns out that we can recover a variant of such a relation by employing our key technical tool, the $\ve$-$\delta$ lemma (Lemma~\ref{lem:epsdelta}). Indeed, combining the approaches of Theorems~\ref{thm:cost_lowerbound} and~\ref{thm:dist_upperbound} gives
\begin{equation}\begin{aligned}
  \norm{\phi^{\otimes \ceil{C n}}}{\mu} & \geq \inf_{X \in \B_{\ve_n}(\rho^{\otimes n})} \norm{X}{\mu} \geq \inf_{\norm{Z}{\mu}\leq 1} d_\h^{2 \delta_n}(\rho^{\otimes n} \| Z) \left( 1 - 2 \delta_n - 2\ve_n \right) \geq \frac{1}{\norm{\phi^{\otimes \floor{D n}}}{\mu}^{\circ}}\left( 1 - 2 \delta_n - 2\ve_n \right) 
\end{aligned}\end{equation}
for all sufficiently large $n$, which then yields the relation
\begin{equation}\begin{aligned}
   \frac{R_{c,\OO_\mu} (\rho)}{-L_{\mu,\infty}^\circ(\phi)} \geq \frac{R_{d,\OO_\mu} (\rho)}{L_{\mu,\infty}(\phi)}.
   \label{eq:no_free_lunch_variant}
\end{aligned}\end{equation}
In this context, the values of $L_{\mu,\infty}(\phi)$ and $-L_{\mu,\infty}^\circ(\phi)$ can be thought of as `renormalisation factors' that account for the specific properties of the norm $\norm{\cdot}{\mu}$ and the choice of the target state $\phi$. 
Since by the Cauchy--Schwartz inequality $\norm{\phi^{\otimes n}}{\mu} \norm{\phi^{\otimes n}}{\mu}^\circ\geq \< \phi^{\otimes n}, \phi^{\otimes n}\> = 1$, we see immediately that $\frac{-L_{\mu,\infty}^\circ(\phi)}{L_{\mu,\infty}(\phi)}\leq 1$; thus, 
the above relation~\eqref{eq:no_free_lunch_variant} would be implied by the no free lunch inequality $R_{c,\OO_\mu} (\rho) \geq R_{d,\OO_\mu} (\rho)$.  
This resembles analogous results in the description of quantum resource transformations under completely positive and trace-preserving maps, where similar renormalisation coefficients are needed~\cite{takagi_2021,kuroiwa_2020}.
\end{remark}

\subsection{Single-letter bounds}\label{subsec:singleletter}

Our aim now is to obtain bounds which are single letter, that is, they do not require regularisation and thus can be more easily evaluated in practice. As a first step towards simplifying the bounds obtained in the previous section, let us notice that, instead of using the parameters $\delta$ (in Theorem~\ref{thm:cost_lowerbound}) or $\ve$ (in Theorem~\ref{thm:dist_upperbound}),  
we could have simply fixed them as zero to begin with. This has the intuitive interpretation of using the zero-error distillation rate to lower bound the resource cost, or vice versa. We formalise this as follows.
\begin{boxed}{white}
\begin{corollary}\label{cor:zero_error_bounds}
For any state $\rho$ and any pure state $\phi$, the rate of transformation from $\phi$ to $\rho$ satisfies
\begin{equation}\begin{aligned}
	r(\phi \toOmu \rho)^{-1} &\geq \frac{\limsup_{n\to\infty} \frac1n \inf_{\norm{Z}{\mu}\leq 1} D^{0}_\h (\rho^{\otimes n} \| Z) }{ \loginf{\phi} }\\
	&= \frac{ \limsup_{n\to\infty} \frac1n \log \sup \lset \< W, \rho^{\otimes n} \> \bar \norm{W}{\mu}^\circ \leq 1,\; \norm{W}{\infty} = \< W, \rho^{\otimes n} \> \rset }{ \loginf{\phi} },
\end{aligned}\end{equation}
and analogously the reverse transformation satisfies
\begin{equation}\begin{aligned}
	r(\rho \toOmu \phi) \leq \frac{\liminf_{n\to\infty} \frac1n  \log \norm{\rho^{\otimes n}}{\mu} }{ -\loginfd{\phi} }.
\end{aligned}\end{equation}
\end{corollary}
\end{boxed}
The evaluation of such bounds is generally non-trivial. It certainly can, however, simplify under additional assumptions. Assume, for instance, that $\norm{X^{\otimes n}}{\mu} \leq \norm{X}{\mu}^n$ which in practice will be satisfied in virtually all cases of interest; we then immediately get
\begin{equation}\begin{aligned}
	R_{d,\O_\mu}(\rho) = r(\rho \toOmu \phi) \leq \frac{ \log \norm{\rho}{\mu} }{ -\loginfd{\phi} }.
\end{aligned}\end{equation}
In entanglement theory, for instance, fixing $\phi$ as the maximally entangled state and choosing $\norm{X}{\mu} = \norm{X^\Gamma}{1}$ gives $\loginfd{\phi} = 1$ and the above recovers one of the oldest known computable upper bounds on distillable entanglement: the logarithmic negativity~\cite{vidal_2002}.

We have thus obtained what could be considered an unsurprising bound: the rate of distillation is upper bounded by the quantity $\log \norm{\rho}{\mu}$, which characterises one-shot, zero-error dilution (cf.~\ref{thm:oneshot_dilution}), and thus clearly should upper bound the asymptotic rate of dilution (as the cost can only increase in the one-shot setting). This natural correspondence does \emph{not} immediately extend to the bound in the other direction: it is not clear whether one-shot, zero-error distillation gives a better bound than asymptotic zero-error distillation. This can be considered as a consequence of the fact that at no point do we assume that the set of operations $\OO_\mu$ is closed under tensor product ---  it might be the case that $\Lambda \in \OO_\mu$, but the $n$-copy protocol $\Lambda^{\otimes n}$ is no longer in $\OO_\mu$.

We will thus need more assumptions in order to bound $D^{0}_\h$. To see this, recall that
\begin{equation}\begin{aligned}
	\inf_{\norm{Z}{\mu} \leq 1} D_\h^0(\rho^{\otimes n} \| Z) &= \log \sup \lset \< W, \rho^{\otimes n} \> \bar \norm{W}{\mu}^\circ \leq 1,\; \norm{W}{\infty} = \< W, \rho^{\otimes n} \> \rset.
\end{aligned}\end{equation}
We would now like to lower bound $\inf_{\norm{Z}{\mu} \leq 1} D_\h^0(\rho^{\otimes n} \| Z)$ with $n \inf_{\norm{Z}{\mu} \leq 1} D_\h^0(\rho \| Z)$ so that the problem reduces to a single-letter quantity. As the operator norm $\norm{\cdot}{\infty}$ is multiplicative under tensor product, such a bound would indeed be possible provided that
\begin{equation}\begin{aligned}\label{eq:subadditive_dual}
	\norm{W^{\otimes n}}{\mu}^\circ \leq \left(\norm{W}{\mu}^\circ\right)^n,
\end{aligned}\end{equation}
so that we can take $W^{\otimes n}$ as a feasible solution. However, for many resource theories of interest the \emph{opposite} inequality is satisfied here: since $\norm{\cdot}{\mu}$ is typically a submultiplicative gauge, this implies the supermultiplicativity of $\norm{\cdot}{\mu}^\circ$ under tensor product. The only way to satisfy~\eqref{eq:subadditive_dual} in such cases is for \emph{equality} to hold there; luckily, such multiplicative norms are available in many theories of interest.

We formalise all of the above as follows.
\begin{boxed}{white}
\begin{corollary}\label{cor:single_letter_bounds}
If the given gauge satisfies $\norm{X^{\otimes n}}{\mu} \leq \norm{X}{\mu}^n$ for all $X$, then
\begin{equation}\begin{aligned}
	r(\rho \toOmu \phi) \leq \frac{ \log \norm{\rho}{\mu} }{ -\loginfd{\phi} }.
\end{aligned}\end{equation}

If the given gauge satisfies $\norm{Y^{\otimes n}}{\mu}^\circ \leq \left(\norm{Y}{\mu}^\circ\right)^n$ for all $Y$, then
\begin{equation}\begin{aligned}
	r(\phi \toOmu \rho)^{-1} &\geq \frac{ \inf_{\norm{Z}{\mu}\leq 1} D^{0}_\h (\rho \| Z) }{ \loginf{\phi} }\\
	&= \frac{ \log \sup \lset \< W, \rho \> \bar \norm{W}{\mu}^\circ \leq 1,\; \norm{W}{\infty} = \< W, \rho \> \rset }{ \loginf{\phi} }.
\end{aligned}\end{equation}
\end{corollary}
\end{boxed}
Here we remark that, as long as the gauge $\norm{\cdot}{\mu}$ is lower semicontinuous in the trace norm topology, then the two conditions $\norm{X^{\otimes n}}{\mu} \leq \norm{X}{\mu}^n$ and $\norm{Y^{\otimes n}}{\mu}^\circ \leq \left(\norm{Y}{\mu}^\circ\right)^n$ can be simultaneously true if and only if $\norm{X^{\otimes n}}{\mu} = \norm{X}{\mu}^n$ and $\norm{Y^{\otimes n}}{\mu}^\circ = \left(\norm{Y}{\mu}^\circ\right)^n$ for all $X$ and $Y$. This is because lower semicontinuity ensures that $\norm{X}{\mu} = \sup \lset \< Y, X \> \bar \norm{Y}{\mu}^\circ \leq 1 \rset$ by the Fenchel--Moreau theorem (see e.g.~\cite[Theorem~2.3.3]{zalinescu_2008} and~\cite[Theorem~15.1]{rockafellar_1970}), 
which then gives
\begin{equation}\begin{aligned}
  \norm{X^{\otimes n}}{\mu} &= \sup \lset \< Y_n, X^{\otimes n} \> \bar \norm{Y_n}{\mu}^\circ \leq 1 \rset  \geq \sup \lset \< Y^{\otimes n}, X^{\otimes n} \> \bar \norm{Y^{\otimes n}}{\mu}^\circ \leq 1 \rset\\
  &\geq \sup \lset \< Y^{\otimes n}, X^{\otimes n} \> \bar \norm{Y}{\mu}^\circ \leq 1 \rset = \norm{X}{\mu}^n
\end{aligned}\end{equation}
as long as $\norm{\cdot}{\mu}^\circ$ is subadditive. The same is clearly true for $\norm{\cdot}{\mu}^\circ$ when $\norm{\cdot}{\mu}$ is subadditive.

In practice, we will often obtain a single-letter lower bound for resource cost in a slightly different manner. Say that a given resource theory is characterised by a gauge $\norm{\cdot}{\mu}$ which is not necessarily multiplicative itself, but there exists a multiplicative gauge $\norm{\cdot}{\gamma}$ such that $\norm{X}{\gamma} \leq \norm{X}{\mu}$ for all $X$. We can then bound
\begin{equation}\begin{aligned}
	\inf_{\norm{Z}{\mu} \leq 1} D_\h^0(\rho^{\otimes n} \| Z) \geq \inf_{\norm{Z}{\gamma} \leq 1} D_\h^0(\rho^{\otimes n} \| Z) \geq n \inf_{\norm{Z}{\gamma} \leq 1} D_\h^0(\rho \| Z)
\end{aligned}\end{equation}
for any $\rho$. We state the resulting lower bound as a separate result, as it will find important applications in several theories of interest.

\begin{boxed}{white}
\begin{corollary}[Multiplicative gauge bound]\label{cor:single_letter_gamma_bound}
Let $\norm{\cdot}{\gamma}$ be a resource gauge such that $\norm{X}{\gamma} \leq \norm{X}{\mu}\, \forall X$ and $\norm{Y^{\otimes n}}{\gamma}^\circ = \left(\norm{Y}{\gamma}^\circ\right)^n \, \forall Y$. Then
\begin{equation}\begin{aligned}
	r(\phi \toOmu \rho)^{-1} &\geq \frac{ \inf_{\norm{Z}{\gamma}\leq 1} D^{0}_\h (\rho \| Z) }{ \loginf{\phi} }\\
	&= \frac{ \log \sup \lset \< W, \rho \> \bar \norm{W}{\gamma}^\circ \leq 1,\; \norm{W}{\infty} = \< W, \rho \> \rset }{ \loginf{\phi} }.
\end{aligned}\end{equation}
\end{corollary}
\end{boxed}
Note the necessity to consider two different gauges, $\norm{\cdot}{\gamma}$ and $\norm{\cdot}{\mu}$, in this bound.

\subsection{Lower bounds on resource cost of mixtures}\label{subsec:mixtures}

One potential drawback of the single-letter bounds of the previous section is that they rely on the quantity $\inf_{\norm{Z}{\mu}\leq 1} D^0_\h(\rho \| Z)$, which will trivialise when the state $\rho$ is full rank. Below, we present an alternative approach to bounding the rate $r(\phi \to \rho)$ that avoids this problem, and can be applied also to full-rank states in finite dimensions. This provides a more direct way of lower bounding the smoothed and regularised quantity encountered in Theorem~\ref{thm:cost_lowerbound}, without going to the inherently single-shot quantity $D^0_\h$.

\begin{boxed}{white}
\begin{theorem} \label{thm:resource_cost_mixtures}
Let $\phi$ be pure, and let $\rho$ be a state decomposable as $\rho = \sum_x p_x \frac{\Pi_x}{d_x}$, where the range of $x$ is finite, for each $x$ the operator $\Pi_x$ is a projector onto a subspace of dimension $d_x$, and $\Pi_x \Pi_{x'} =0$ for $x\neq x'$. Consider a gauge $\norm{\cdot}{\gamma}$ such that $\norm{X}{\gamma} \leq \norm{X}{\mu} \, \forall X$ and $\norm{Y \otimes Z}{\gamma}^\circ \leq \norm{Y}{\gamma}^\circ \norm{Z}{\gamma}^\circ \, \forall Y,Z$. Then 
\bb
r\left(\phi \toOmu \rho \right)^{-1} \geq \frac{ \sum_x p_x \log \frac{1}{\norm{\Pi_x}{\gamma}^\circ} - H(p) }{ \loginf{\phi} },
\label{resource_cost_mixtures}
\ee
where $H(p)\coloneqq -\sum_x p_x \log p_x$ is the entropy of $p$.
\end{theorem}
\end{boxed}

The assumption that the range of $x$ is finite is necessary to apply the typicality arguments in what follows.

\begin{proof}
Since we intend to apply Theorem~\ref{thm:cost_lowerbound}, the first step is to lower bound 
\bb
\sup \lset \frac{1}{\norm{Q}{\mu}^\circ} \bar \< Q, \rho^{\otimes n} \> \geq 1-\delta,\; \norm{Q}{\infty} \leq 1 \rset
\ee
for large values of $n$. For some fixed (small) parameter $\eta>0$, define the strongly typical subspace of sequences $x^n = x_1\ldots x_n$ by setting~\cite[Definition~14.7.2]{wilde_2017}
\bb
T_\eta \coloneqq \lset x^n \bar \left|\frac1n N(x|x^n) - p_x \right| \leq \eta\, (1-\delta_{0,p_x})\; \forall x \rset ,
\ee
where $\delta_{0,q}=1$ if $q=0$ and $\delta_{0,q}=0$ otherwise, and $N(x|x^n)$ denotes the number of times the symbol $x$ appears in the sequence $x^n$. It is well known that~\cite[Eq.~(14.72) and~(14.73)]{wilde_2017}
\begin{align}
\lim_{n\to\infty} \sum_{x^n \in T_\eta} p_{x^n} &= 1 , \label{typical_set_probability} \\
|T_\eta| &\leq 2^{n (H(p) + c\eta)} , \label{typical_set_size}
\end{align}
where $p_{x^n}\coloneqq \prod_{i=1}^n p_{x_i}$, and $c\coloneqq - \sum_{x:\, p_x>0} \log p_x > 0$ is a positive constant. Denoting analogously $\Pi_{x^n}\coloneqq \bigotimes_{i=1}^n \Pi_{x_i}$ and $d_{x^n} \coloneqq \prod_{i=1}^n d_{x_i}$, let us now define
\bb
Q_n \coloneqq \sum_{x^n \in T_\eta} \Pi_{x^n} .
\ee
On the one hand, since $Q_n$ is an orthogonal projector we have that $\norm{Q_n}{\infty}=1$; on the other, for every $\delta>0$ and for sufficiently large $n$ it holds that
\bb
\< Q_n, \rho^{\otimes n} \> = \sum_{x^n} p_{x_n} \<Q_n, \frac{\Pi_{x^n}}{d_{x^n}} \> = \sum_{x^n\in T_\eta} p_{x_n} \geq 1-\delta ,
\ee
Now, we estimate
\bb
\norm{Q_n}{\mu}^\circ &\textleq{(i)} \norm{Q_n}{\gamma}^\circ \\
&\leq \sum_{x^n\in T_\eta} \norm{\Pi_{x^n}}{\gamma}^\circ \\
&\textleq{(ii)} \sum_{x^n\in T_\eta} \prod_x \left(\norm{\Pi_{x}}{\gamma}^\circ\right)^{N(x|x^n)} \\
&\textleq{(iii)} \sum_{x^n\in T_\eta} \prod_x \left(\norm{\Pi_{x}}{\gamma}^\circ\right)^{n (p_x + \eta)} \\
&= |T_\eta| \prod_x \left(\norm{\Pi_{x}}{\gamma}^\circ\right)^{n (p_x + \eta)} \\
&\textleq{(iv)} 2^{n\left(H(p) + c\eta + \sum_x \left(p_x \log \norm{\Pi_x}{\gamma}^\circ + \eta\, \left| \log \norm{\Pi_x}{\gamma}^\circ \right| \right) \right)} .
\ee
Here, (i)~follows from the inequality $\norm{\cdot}{\gamma}\leq \norm{\cdot}{\mu}$, which implies by duality that $\norm{\cdot}{\gamma}^\circ \geq \norm{\cdot}{\mu}^\circ$; (ii)~comes from the sub-multiplicativity of $\norm{\cdot}{\gamma}^\circ$; (iii)~descends from the definition of strongly typical set; and (iv)~is just an application of~\eqref{typical_set_size}.

Thanks to Theorem~\ref{thm:cost_lowerbound}, from the above chain of inequalities we deduce that
\bb
r(\phi \toOmu \rho)^{-1} &\geq \frac{1}{\loginf{\phi}} \lim_{\delta \to 0} \limsup_{n\to\infty} \frac{1}{n} \log \sup \lset \frac{1}{\norm{Q}{\mu}^\circ} \bar \< Q, \rho^{\otimes n} \> \geq 1-\delta,\; \norm{Q}{\infty} \leq 1 \rset \\
&\geq \frac{1}{\loginf{\phi}} \limsup_{n\to\infty} \frac1n \log \frac{1}{\norm{Q_n}{\mu}^\circ} \\
&\geq \frac{ - \sum_x p_x \log \norm{\Pi_x}{\gamma}^\circ - H(p) - \eta \left( c + 
\sum_x \left| \log \norm{\Pi_x}{\gamma}^\circ \right| \right) }{ \loginf{\phi} }.
\ee
Since $\eta>0$ was arbitrary, we can take $\eta\to 0^+$ and complete the proof.
\end{proof}

\begin{remark}
A closer look to the proof of Theorem~\ref{thm:resource_cost_mixtures} shows that we in fact proved the slightly stronger bound
\bb
r\left(\phi \toOmu \rho \right)^{-1} &\geq \frac{1}{\loginf{\phi}} \limsup_{n\to\infty} \left(- \frac1n\, \log \norm{\sumno_{x\in T_\eta} \Pi_{x^n}}{\mu}^\circ \right) ,
\label{resource_cost_mixtures_better}
\ee
from which one deduces~\eqref{resource_cost_mixtures} by using the triangle inequality and replacing $\norm{\cdot}{\mu}^\circ$ with $\norm{\cdot}{\gamma}^\circ$.
\end{remark}

\subsection{On positivity, and comparison with other gauge-based bounds}\label{subsec:positivity}

Here we take a short intermezzo to highlight the importance of making an appropriate choice of the gauge function $\norm{\cdot}{\mu}$ to be used in the bounds of this work. We do this through a comparison with a different approach, which traces its conceptual beginnings to works in entanglement theory by Rains~\cite{rains_2001}, and which found use recently in the establishment of several norm-based bounds in the theory of entanglement~\cite{wang_2016,wang_2017-3,wang_2017-1} and magic~\cite{wang_2020}.

Although previous works did not explicitly formalise it in this way, their approach can be thought of as studying resource transformations under a class of completely positive maps defined as 
\begin{equation}\begin{aligned}
	\O_{\mu+} \coloneqq \lset \Lambda \bar \Lambda \in \mathrm{CP}, \; \norm{\Lambda(X)}{1} \leq \norm{X},\; \norm{\Lambda(X)}{\mu} \leq \norm{X}{\mu} \rset.
\end{aligned}\end{equation}
This means that $\O_{\mu+} \subset \O_{\mu}$, and hence it appears that this choice of operations should never provide better bounds through the approach that we have employed in our work (cf.\ the discussion in Section~\ref{sec:intuition}). In particular, our lower bounds on resource cost (Theorem~\ref{thm:cost_lowerbound} and Corollary~\ref{cor:single_letter_bounds}) provide a strict improvement over lower bounds on entanglement cost which were established in~\cite{wang_2017-3,wang_2017-1}.

It might then come as a surprise that a na\"ive application of the upper bound for distillable resource that we have found (Theorem~\ref{thm:dist_upperbound} and Corollary~\ref{cor:single_letter_bounds}) can actually perform \emph{worse} than corresponding bounds derived in~\cite{wang_2016,wang_2020}. The reason for this is that the bounds in the latter works make an explicit use of the positivity of operations $\O_{\mu+}$, which in this case results in tighter limitations. But the same result can be recovered in our approach, as we now show.

Define the function $p_\mu(X)$ as a `positive' variant of the gauge $\norm{\cdot}{\mu}$:
\begin{equation}\begin{aligned}\label{eq:positive_norm}
	p_\mu(X) \coloneqq& \inf \lset \norm{Z}{\mu} \bar X \leq Z \rset.
\end{aligned}\end{equation}
Note that $p_{\mu}(X)\leq\norm{X}{\mu}$. 
Although useful, this quantity is not a resource gauge per se, because we require our gauges to be symmetric. We can, however, define a bona fide resource gauge by simply symmetrising the above, that is,
\begin{equation}\begin{aligned}
	\norm{X}{\mu+} \coloneqq \max \{ p_\mu(X),\, p_\mu(-X) \}.
\end{aligned}\end{equation}
Since every resource gauge can be written as $\norm{Z}{\mu} = \inf \lset \lambda \bar Z = \lambda U,\; \norm{U}{\mu} \leq 1 \rset$, we can rewrite $\norm{X}{\mu+}$ as
\begin{equation}\begin{aligned}
	\norm{X}{\mu+} &= \inf \lset \max \{ \lambda_+,\, \lambda_- \} \bar -\lambda_- Q \leq X \leq \lambda_+ Z,\; \norm{Q}{\mu}, \norm{Z}{\mu} \leq 1 \rset\\
&= \inf \lset \lambda \bar -\lambda Q \leq X \leq \lambda Z,\; \norm{Q}{\mu}, \norm{Z}{\mu} \leq 1 \rset.
\end{aligned}\end{equation}
It follows by construction that $\norm{\cdot}{\mu+}$ is in fact a monotone under $\O_{\mu+}$, allowing us to immediately apply our results.

Importantly, as long as $X \geq 0$, it holds that $\norm{X}{\mu+} = p_\mu(X)$. Applying the bounds of Theorem~\ref{thm:dist_upperbound} and Corollary~\ref{cor:single_letter_bounds} and making use of the fact that $\O_{\mu+}$ are positive maps --- meaning that, for any $X \geq 0$, $\Lambda(X)$ will always be a positive semidefinite operator --- we obtain the following result.

\begin{boxed}{white}
\begin{corollary}\label{cor:upper_bound_dist_positive}
For any state $\rho$ and any pure state $\phi$, the rate of transformation from $\rho$ to $\phi$ under completely positive operations in $\O_\mu$ satisfies
\begin{equation}\begin{aligned}
	r(\rho \toOx{\mu+} \phi) &\leq \frac{\lim_{\ve \to 0} \liminf_{n\to\infty} \frac1n \inf_{X \in B_{\ve}(\rho^{\otimes n})\cap\T_+(\H)} \log p_\mu(X) }{ -\loginfd{\phi} }\\
	&\leq \frac{\liminf_{n \to \infty} \frac1n \log p_\mu\left(\rho^{\otimes n}\right)}{-\loginfd{\phi}}.
\end{aligned}\end{equation}
As long as the gauge $\norm{\cdot}{\mu}$ is sub-multiplicative, i.e.\ $\norm{X \otimes Y}{\mu} \leq \norm{X}{\mu} \norm{Y}{\mu}$, this gives
\begin{equation}\begin{aligned}
	r(\rho \toOx{\mu+} \phi) \leq \frac{\log p_\mu(\rho)}{-\loginfd{\phi}} = \frac{\log \norm{\rho}{\mu+}}{-\loginfd{\phi}}.
\end{aligned}\end{equation}
\end{corollary}
\end{boxed}
This in fact recovers the upper bound on distillable entanglement $E_W$ found in~\cite{wang_2016} and an upper bound on magic state distillation $\theta_{\max}$ from~\cite{wang_2020} as special cases. Since these are some of the best known computable upper bounds in their corresponding resource theories, we thus see that our results can recover the leading constraints on resource distillation, as long as care is taken to choose a suitable gauge function in the application of our bounds.

It is rather interesting to note that a similar `positivity' assumption does not appear to be helpful in improving our bounds on the reverse transformation $\phi \to \rho$. On the one hand, replacing a gauge $\norm{\cdot}{\mu}$ with $\norm{\cdot}{\mu+}$ will, in general, sacrifice its multiplicativity, which we need to apply the single-letter bounds in Corollary~\ref{cor:single_letter_bounds} or~\ref{cor:single_letter_gamma_bound}; on the other hand, since the quantity $D^\ve_\h$ that we employ is a maximisation, explicitly adding a positivity assumption there would only lead to a worse bound (and would, in fact, reduce to the previous bounds of~\cite{wang_2017-3,wang_2017-1}). This can be considered as the reason why our approach allowed us to improve on previous bounds: they all took positivity for granted, while it appears that the characterisation of resource dilution benefits from going beyond positive operators.

\section{Applying the results to specific resource theories}\label{sec:apps}

The basic ingredient of any resource theory of quantum states is the set $\FF$ of states which are designated as free. This set immediately gives rise to several different gauges that can be employed in the framework of this work.

\begin{table*}
\setlength{\extrarowheight}{5pt}
\makebox[\textwidth]{%
  \begin{tabular}{@{}>{\raggedright}p{2.7cm} >{\raggedright}p{3.2cm} >{\raggedright}p{2cm} 
  p{1.35cm} p{1.5cm} p{1.6cm} p{1.55cm} p{1.7cm}}
    \toprule
        \bfseries Resource theory & \bfseries Gauge $\norm{X}{\mu}$ & $\norm{X\otimes Y}{\mu}^\circ \!\stackrel{?}{=}\! \norm{X}{\mu}^\circ \norm{Y}{\mu}^\circ$ 
        & \bfseries Target state $\ket\phi$ & $\norm{\phi^{\otimes n}}{\mu}$ & $1/\norm{\phi^{\otimes n}}{\mu}^\circ$ & $\loginf{\phi}$ &  $-\loginfd{\phi}$ \\
    Entanglement & base norm $\norm{X}{\SEP}$ & No &  $\ket{\Phi_2}$ & $2^{n+1}-1$ & $2^n$ & 1 & 1 \\
   ~\cite{horodecki_2009} & Negativity $\norm{X^\Gamma}{1}$ & Yes &  $\ket{\Phi_2}$ & $2^n$ & $2^n$ & 1 & 1 \\
    & Reshuffled negativity $\norm{X^\R}{1}$ & Yes &  $\ket{\Phi_2}$ & $2^n$ & $2^n$ & 1 & 1 \\[1.2ex]
    Qudit magic & base norm $\norm{X}{\FF_W}$ & No &  $\ket{S}$ & $\leq 3^n$ & $2^n$ & $\leq \log 3$ & 1\\
   ~\cite{veitch_2014} &&& $\ket{N}$ & $\leq 2^n$ & $\left(\frac32\right)^{n}$ & $\leq \frac12 \log\frac{11}{3}$ & $\log \frac32$\\[-10pt]
    &&& $\ket{H_+}$ & $\leq\! \left(\frac{3+\sqrt{3}}{3}\right)^{\!n}$ & $(3-\sqrt{3})^n$ & $\leq\!\log \frac{3+\sqrt{3}}{3}$ & $\log (3\!-\!\sqrt{3})$ \\[-10pt]
    & Wigner negativity $\norm{X}{W}$ & Yes & $\ket{S}$ & $\left(\frac{5}{3}\right)^n$ & $1$ & $\log \frac53$ & $0$\\[-10pt]
    &&& $\ket{N}$ & $\left(\frac{5}{3}\right)^n$ & $1$ & $\log \frac53$ & $0$ \\[-10pt]
    &&& $\ket{H_+}$ & $\left(\frac{1+2\sqrt{3}}{3}\right)^n$ & $1$ & $\log \frac{1+2\sqrt{3}}{3}$ & $0$\\[-10pt]
    Qubit magic & base norm $\norm{X}{\STAB}$ & No & $\ket{T}$ & $\leq \sqrt{2}^n$ & $4 - 2\sqrt{2}$ & $\leq\! \log 1.29$ & $\log (4 -\! 2\sqrt{2})$ \\
   ~\cite{veitch_2014,howard_2017} &&& $\ket{\mathrm{Hog}}$ & $\leq \left(\frac{19}{5}\right)^n$ & $\left(\frac{12}{5}\right)^n$ & $\leq \log \frac{19}{5}$ & $\log \frac{12}{5}$ \\[-10pt]
    & Stabiliser norm $\norm{\cdot}{\mathcal{P}}$ & Yes & $\ket{T}$ & $\left(\frac{1+\sqrt{2}}{2}\right)^n$ & $1$ & $\log \frac{1+\sqrt{2}}{2}$ & $0$ \\[-10pt]
    &&& $\ket{\mathrm{Hog}}$ & $\left(\frac{11}{4}\right)^n$ & $1$ & $\log \frac{11}{4}$ & $0$ \\[-10pt]
    \bottomrule
    \end{tabular}
 }%
 \caption{
 \textbf{Survey of resources, gauge-based monotones, and choices of reference states $\boldsymbol{\ket{\phi}}$.}
 We overview the formalism of this work applied to the two considered resource theories, namely, quantum entanglement and magic-state quantum computation. For each resource, we give a selection of norm-based measures that can be used in our framework. We specify whether a given norm has a multiplicative dual (so that Corollary~\ref{cor:single_letter_gamma_bound} can be applied). We then give estimates or, where known, exact values of the parameters $\loginf{\phi}$ and $\loginfd{\phi}$ that are required for the application of Theorems~\ref{thm:cost_lowerbound} and~\ref{thm:dist_upperbound}. For the theory of magic, the upper bounds on $\norm{\phi^{\otimes n}}{\mu}$ are based on the evaluation of the single-letter quantity $\norm{\phi}{\mu}$ (see~\cite{veitch_2014,howard_2017,bravyi_2019,takagi_2021}) and using the sub-multiplicativity of the norm; for $\loginf[\FF_W]{\proj{N}}$, we also give an improved estimate computed in Section~\ref{sec:app_magic}, while for, $\loginf[\STAB]{\proj{T}}$, we use a many-copy estimate that was obtained in~\cite{heinrich_2019}. The evaluation of $\norm{\phi^{\otimes n}}{\mu}^\circ$ proceeds by first computing $\norm{\phi}{\mu}^\circ$ (see~\cite{bravyi_2019,wang_2020}), and then using the fact that, for the considered cases, the dual norms are multiplicative~\cite{bravyi_2019,wang_2020}.
 See Sections~\ref{sec:app_ent} and~\ref{sec:app_magic} for more details on the choices of target states and on the computation of the quantities.
 }
\label{tab:resources}
\end{table*}

A basic assumption we will make is that $\FF$ is a closed and convex set. In finite-dimensional theories, as long as $\mathrm{span}(\FF) = \T_\sa(\H)$, i.e.\ the free states are a full-measure subset of all quantum states, then a natural norm that can be defined is the \deff{base norm}~\cite{jameson_1970, lami_2018-1}:
\begin{equation}\begin{aligned}\label{eq:def_base_norm}
	\norm{X}{\FF} \coloneqq \inf \lset \lambda_+ + \lambda_- \bar X = \lambda_+ \sigma_+ - \lambda_- \sigma_-,\; \sigma_{\pm} \in \FF,\; \lambda_\pm \in \RR_+ \rset.
\end{aligned}\end{equation}
The condition on the span of $\FF$ is satisfied in resource theories such as quantum entanglement or magic. Even when $\mathrm{span}(\FF) \neq \T_\sa(\H)$, which happens in particular for several infinite-dimensional resource theories of operational interest, the quantity in~\eqref{eq:def_base_norm} can still be defined, although care needs to be taken as it will be infinite for some states~\cite{lami_2021, regula_2021}. 
This quantity is directly related to a commonly used resource monotone called the \deff{standard robustness}, defined as $R_\FF^s(\rho) \coloneqq \inf \lset \lambda \in \RR_+ \bar \rho + \lambda \sigma_- = (1+\lambda) \sigma_+,\; \sigma_\pm \in \FF \rset$. It is not difficult to notice that, for a quantum state, we have $\norm{\rho}{\FF} = 1 + 2 R^s_\FF(\rho)$.

Another type of a norm can be defined if the free states satisfy that $\FF = \mathrm{conv} \lset \proj{\psi} \bar \ket{\psi} \in \V \rset$ for some set $\V$ defined in the underlying Hilbert space $\mathcal{H}$, such that $\mathrm{span}(\V) = \mathcal{H}$. This is the case for resource theories such as entanglement, magic, or quantum coherence. We can then define
\begin{equation}\begin{aligned}
	\norm{X}{\V} \coloneqq \inf \lset \sum_i |\lambda_i| \bar X = \sum_i \lambda_i \ket{v_i}\!\bra{w_i},\; \lambda_i \in \CC,\; \ket{v_i}, \ket{w_i} \in \V \rset.
\end{aligned}\end{equation}
This gives rise to the projective tensor norm in entanglement theory~\cite{rudolph_2001,rudolph_2005}, the $\ell_1$ norm in the resource theory of coherence~\cite{baumgratz_2014}, and the so-called dyadic negativity in magic theory~\cite{seddon_2021}. Note that this norm can be defined even for linear operators which are not self-adjoint.

Yet another form of a norm can be defined based on another resource measure, the \deff{generalised robustness} $R^g_\FF(\rho) \coloneqq \inf \lset \lambda \bar \rho + \lambda \omega = (1+\lambda) \sigma,\; \sigma \in \FF,\; \omega \in \D(\H) \rset$, as follows. The function
\begin{equation}\begin{aligned}
	\norm{X}{g,\FF} \coloneqq \inf \lset \lambda \bar -\lambda \sigma_- \leq X \leq \lambda\sigma_+,\; \sigma_\pm \in \FF \rset
\label{eq:norm_g_FF}
\end{aligned}\end{equation}
can be seen to be a norm in the case of finite-dimensional theories as soon as $\FF$ contains at least one state of full rank, which is satisfied in virtually all cases of interest. For any state, it holds that $\norm{\rho}{g,\FF} = 1 + R^g_\FF(\rho)$.

A useful property of the above norms is that, for a pure state $\phi$, their dual norms all coincide~\cite{regula_2018}:
\begin{equation}\begin{aligned}
	\norm{\phi}{\FF}^\circ = \norm{\phi}{\V}^\circ = \norm{\phi}{g,\FF}^\circ = \sup \lset \< \phi, \sigma \> \bar \sigma \in \FF \rset,
\end{aligned}\end{equation}
where we assumed that each of the norms is well defined. Since we have already seen that evaluating the dual norm on the target state $\phi$ in distillation and dilution protocols is a crucial ingredient of many of our bounds, this insight can help in the applications of the different norms.

The above are just examples of norms (or gauges) that can be used in our approach. The crucial point to consider in any resource theory is how to obtain single-letter bounds on asymptotic transformations rates, leading to efficiently computable bounds. To clarify the applicability of our single-letter results (Corollaries~\ref{cor:single_letter_bounds},~\ref{cor:single_letter_gamma_bound}, and~\ref{cor:upper_bound_dist_positive}), we now summarise how they may be applied in a given resource theory.

\begin{boxed}{white}
Consider a general resource theory whose sets of free states $\FF$ are convex and closed. Assume that $\sigma \in \FF \Rightarrow \sigma^{\otimes n} \in \FF \; \forall n \in \mathbb{N}$, 
which is a weak assumption satisfied in virtually all theories of interest. Consider then any class of quantum channels $\OO$ under which the base gauge  
$\norm{\cdot}{\FF}$ is contractive --- this includes, in particular, the set $\Omax$ of resource non-generating quantum channels, and all subsets thereof. Then:
\begin{enumerate}[(i)]
\item In order to upper bound the distillation rate $r(\rho \toO \phi)$, one needs to compute or upper bound the regularised quantity $\loginfd[\FF]{\phi}$. Then, Corollary~\ref{cor:single_letter_bounds} or~\ref{cor:upper_bound_dist_positive} can be used to establish that
\begin{equation}\begin{aligned}
r(\rho \toO \phi) \leq \frac{\log \norm{\rho}{g,\FF}}{- \loginfd[\FF]{\phi} } \leq \frac{\log \norm{\rho}{\FF}}{- \loginfd[\FF]{\phi} }.
\label{eq:handy_upper_distillable}
\end{aligned}\end{equation}
\item In order to upper bound the dilution rate $r(\phi \toO \rho)$, one needs to compute or upper bound the regularised quantity $\loginf[\FF]{\phi} $, as well as find a \emph{multiplicative} gauge $\norm{\cdot}{\gamma}$ such that $\norm{\cdot}{\gamma} \leq \norm{\cdot}{\FF}$.  Corollary~\ref{cor:single_letter_gamma_bound} then gives
\begin{equation}\begin{aligned}
r(\phi \toO \rho)^{-1} &\geq \frac{ \inf_{\norm{Z}{\gamma}\leq 1} D^{0}_\h (\rho \| Z) }{ \loginf[\FF]{\phi} }.
\label{eq:hand_lower_cost}
\end{aligned}\end{equation}
An alternative, but more tailored approach that might only be applicable to specific quantum states is provided in Theorem~\ref{thm:resource_cost_mixtures}.
\end{enumerate}
\end{boxed}

We discuss these problems in specific resource theories. We also present an overview of different gauge choices for representative resource theories in Table~\ref{tab:resources}.

\section{Application: quantum entanglement}\label{sec:app_ent}

In the manipulation of quantum entanglement, the relevant set of free states is separable states, which can be defined on a bipartite quantum system $AB$ with Hilbert space $\H_{AB} = \H_A\otimes \H_B$ by
\bb
\F = \SEP \coloneqq \operatorname{cl} \operatorname{conv} \left\{ \ketbra{\psi}_A \otimes \ketbra{\phi}_B \right\} ,
\ee
where the closure is with respect to the trace norm, and the convex hull ranges over all normalised pure states $\ket{\psi}_A\in \H_A$ and $\ket{\phi}_B\in \H_B$. Every state $\sigma_{AB} \in \SEP$ can thus be written as~\cite{holevo_2005-1}
\begin{equation}
    \sigma_{AB} = \int \ketbra{\psi}_A \otimes \ketbra{\phi}_B\, \mathrm{d}\mu(\psi,\phi)\, ,
    \label{separable}
\end{equation}
where $\mu$ is a Borel probability measure on the product of the sets of local (normalised) pure states. The above integral is to be interpreted as a Bochner integral in either the space of trace class operators or that of Hilbert--Schmidt operators --- both are well defined as
\begin{equation}
\int \left\|\ketbra{\psi}_A \otimes \ketbra{\phi}_B\right\|_1\, \mathrm{d}\mu(\psi,\phi) = \int \left\|\ketbra{\psi}_A \otimes \ketbra{\phi}_B\right\|_2\, \mathrm{d}\mu(\psi,\phi) = 1
\end{equation}
is finite.\footnote{Observe also that the integrand in~\eqref{separable} is clearly a strongly measurable function, as it is both weakly measurable, since it is continuous on the space of product states equipped with the standard trace norm topology, and almost surely separably valued, since it takes on values in the separable space of trace class operators on the separable Hilbert space $\H$.} 
Conversely, every state of the form~\eqref{separable} is clearly separable.

The relevant target state $\phi$ here is the two-qubit maximally entangled state $\phi = \Phi_2$, where $\ket{\Phi_2} = \frac{1}{\sqrt{2}} (\ket{00}+\ket{11})$; we will refer to the corresponding rates of transformations as distillable entanglement $E_{d,\O} (\rho) \coloneqq r(\rho \toO \Phi_2)$ and entanglement cost $E_{c,\O} (\rho) \coloneqq r(\Phi_2 \toO \rho)^{-1}$.

As for the class $\O$, perhaps the best motivated from an operational standpoint is that of local operations and classical communication, denoted by $\LOCC$~\cite{chitambar_2014}. Other useful classes are $\Omax \coloneqq \NE$ (non-entangling, a.k.a.\ separability-preserving operations~\cite{brandao_2008-1, lami_2021-1}), $\SEPO$ (separable channels)~\cite{rains_1997}, and $\PPTO$ (channels $\Lambda:AB\to A'B'$ such that $\Gamma_{A'}\circ \Lambda \circ \Gamma_A$ is CPTP, where $\Gamma$ denotes the partial transpose)~\cite{rains_2001}.

\subsection{Basic bounds on entanglement cost and distillable entanglement}

The set $\SEP$ immediately gives rise to the base gauge $\norm{X}{\SEP}$, which we recall as
\begin{equation}\begin{aligned}
	\norm{X}{\SEP} \coloneqq \inf \lset \lambda_+ + \lambda_- \bar X = \lambda_+ \sigma_+ - \lambda_- \sigma_-,\; \sigma_{\pm} \in \SEP \rset.
\end{aligned}\end{equation}
As we mentioned before, for normalised quantum states this quantity is equivalent to the measure known as the (standard) robustness of entanglement~\cite{vidal_1999}, in the sense that $\norm{\rho}{\SEP} = 1+2R^s_\SEP(\rho)$.

Let us now discuss the asymptotic properties of this gauge. Using the fact that $\sigma, \tau \in \SEP \Rightarrow \sigma \otimes \tau \in \SEP$, it can be shown that $\norm{X^{\otimes n}}{\SEP} \leq \norm{X}{\SEP}^n$, which immediately gives the bound
\begin{equation}\begin{aligned}
	\loginf[\SEP]{\phi} = \lim_{n\to\infty} \frac1n \log \norm{\phi^{\otimes n}}{\SEP} \leq \log \norm{\phi}{\SEP}
\end{aligned}\end{equation}
for any pure state $\phi$. However, a much tighter bound can be obtained for pure states, where it was shown~\cite{vidal_1999,lami_2021} that $1 + R^s_\SEP(\phi) = \sum_{i=1}^{\infty} \alpha_i$, where $\{\alpha_i\} \in \RR_+$ are the Schmidt coefficients of a given state, i.e.\ $\ket{\phi} = \sum_{i=1}^{\infty} \alpha_i \ket{e_i f_i}$ for some orthonormal bases $\{\ket{e_i}\}_i, \{\ket{f_j}\}_j$. This means that $1 + R^s_\SEP(\phi^{\otimes n}) = \left(1+R^s_\SEP(\phi)\right)^n$ for any pure $\phi$, and hence
\begin{equation}\begin{aligned}\label{eq:asymp_std_rob_for_pure}
  \loginf[\SEP]{\phi} &= \lim_{n\to\infty} \frac1n \log \left(1+2 R^s_\SEP(\phi^{\otimes n})\right)\\  &= \lim_{n\to\infty} \frac1n \log \left(1+R^s_\SEP(\phi^{\otimes n})\right)\\
 	&= \log \left(1 + R^s_\SEP(\phi)\right).
 \end{aligned}\end{equation}
For the relevant choice of the two-qubit maximally entangled state as the target, this gives $\loginf[\SEP]{\Phi_2} = 1$.

The dual of the gauge here is the overlap $\norm{X}{\SEP}^\circ = \sup_{\sigma \in \SEP} \left|\<X,\sigma\>\right|$. It is known~\cite{shimony_1995} that $\norm{\phi}{\SEP}^\circ = \left(\sup_i \alpha_i\right)^{-1}$ for any pure state $\phi$, which in particular means that $\norm{\phi^{\otimes n}}{\SEP}^\circ = \left(\norm{\phi}{\SEP}^\circ \right)^n$, yielding
\begin{equation}\begin{aligned}
	\loginfd[\SEP]{\phi} = \lim_{n \to \infty} \frac1n \log \norm{\phi^{\otimes n}}{\SEP}^\circ = \log \norm{\phi}{\SEP}^\circ.
\end{aligned}\end{equation}
For the two-qubit maximally entangled state, we once again have that $-\loginfd[\SEP]{\Phi_2} = 1$.

We now focus on evaluating the asymptotic bounds established in our work. The physically relevant sets of free operations in this theory all form subsets of the maximal set of free channels, $\Omax$, which 
comprises all non-entangling maps: $\Omax = \big\{ \Lambda \in \CPTP \;|\; \Lambda(\sigma) \in \SEP \; \forall \sigma \in \SEP \big\}$. By choosing any gauge which 
contracts under 
those maps --- such as the base gauge $\norm{\cdot}{\SEP}$ or the  gauge $\norm{\cdot}{g,\SEP}$ based on the generalised robustness 
--- we can use the properties of the resource gauges to constrain the operational properties of entanglement.

Although it is not known what the regularised bound of Theorem~\ref{thm:dist_upperbound} evaluates to, the single-letter bound in Corollary~\ref{cor:single_letter_bounds},
\begin{equation}\begin{aligned}
	E_{d,\O_\SEP} (\rho) \leq \log \left(1 + 2 R^s_\SEP(\rho)\right),
\end{aligned}\end{equation}
is known~\cite{vidal_2002}. As for the lower bound on entanglement cost, the generalised hypothesis testing relative entropy corresponds to a quantity introduced in~\cite{lami_2021-1} as the `tempered robustness' $R_\SEP^\tau$:
\begin{equation}\begin{aligned}
	\inf_{\norm{Z}{\SEP}\leq 1} D^{0}_\h (\rho \| Z)  &= \log \sup \lset \< W, \rho \> \bar \norm{W}{\SEP}^\circ \leq 1,\; \norm{W}{\infty} = \< W, \rho \> \rset
	\\&= 1 + 2 R^\tau_\SEP(\rho).
\end{aligned}\end{equation}
Our zero-error bound in Corollary~\ref{cor:zero_error_bounds} then recovers a lower bound on entanglement cost found in~\cite[Theorem~S7]{lami_2021-1}, although the regularised bound in Theorem~\ref{thm:cost_lowerbound} has a strong potential to improve on that result, provided that it can be computed exactly.

\subsection{Bounds from partial transposition} \label{subsec:partial_transposition}

The issue we encounter is that the base gauge $\norm{\cdot}{\SEP}$ is not multiplicative,\footnote{Pure states already constitute a counterexample to multiplicativity: since $\|\phi\|_{\SEP} = 2\sum_i \alpha_i - 1$ for all pure states $\ket{\phi}$ with Schmidt coefficients $\{\alpha_i\}_i$~\cite{vidal_1999}, it is $\frac12 \big(1+\|\phi\|_{\SEP}\big)$ that is multiplicative, not $\|\phi\|_{\SEP}$. In this simple case, however, we can efficiently understand the asymptotics of $\norm{\phi^{\otimes n}}{\SEP}$. A stronger counterexample is presented in~\cite[Sec.~V.B]{vollbrecht_2001}. With the notation in that paper, and with $\rho_-$ representing the `antisymmetric state' in local dimension $d$, it can be verified that $\|\rho_-\|_{\SEP} = 3$ but $\big\|\rho_-^{\otimes 2}\big\|_{\SEP} \leq (3d+1)/(d-1)$ (in fact, with equality) because $\rho_-^{\otimes 2} = \big(2d\rho_\# - (d+1)\rho_+^{\otimes 2}\big)/(d-1)$, with both $\rho_\#$ and $\rho_+^{\otimes 2}$ being separable. Here the value of $\norm{\rho_-^{\otimes n}}{\SEP}$ for large $n$ is still an open problem.} preventing an efficient evaluation of the lower bounds on $E_{c,\O_\SEP}$, and in particular the computation of the single-letter results. We therefore consider a relaxation: instead of the set $\FF = \SEP$, we consider as our set of interest the set of all states with a positive partial transpose, or in fact an even larger set: the set $\wt\PPT = \lset X \bar X^\Gamma \geq 0,\; \Tr X =1 \rset$. The base gauge $\norm{\cdot}{\wt\PPT}$ is then simply the entanglement negativity, $\norm{X}{\wt\PPT} = \norm{X^\Gamma}{1}$, which is straightforwardly verified to be multiplicative. By the inclusion $\SEP \subseteq \wt\PPT$ it holds that $\norm{X^\Gamma}{1} \leq \norm{X}{\SEP}$, which allows us to apply Corollary~\ref{cor:single_letter_gamma_bound} to give
\begin{equation}\begin{aligned}\label{eq:temp_neg_lower_bound}
	E_{c,\O_{\SEP}} (\rho) &\geq \inf_{\norm{Z^\Gamma}{1}\leq 1} D^{0}_\h (\rho \| Z)\\
	&=\log \max \lset \< W, \rho \> \bar \norm{W^\Gamma}{\infty} \leq 1,\; \norm{W}{\infty} = \< W, \rho \> \rset.
\end{aligned}\end{equation}
This quantity exactly equals the \deff{tempered logarithmic negativity} $E_\tau$ of~\cite{lami_2021-1}. We have thus not only recovered this single-letter bound on entanglement cost, but in fact endowed it with an operational meaning: using Theorem~\ref{thm:oneshot_distillation}, we have that
\begin{equation}\begin{aligned}
	E_{d,\O_{\wt\PPT}}^{(1),\textrm{exact}} (\rho) = \floor{E_\tau(\rho)};
\end{aligned}\end{equation}
in other words, $E_\tau(\rho)$ quantifies exactly the distillable entanglement under linear maps which contract the trace norm and the negativity.

\subsection{Bounds from reshuffling criterion}

We will now employ Corollary~\ref{cor:single_letter_gamma_bound} to deduce a novel lower bound on the entanglement cost. To construct it, let us consider the \deff{reshuffling} operation on $\T(\H_{AB})$, the Banach space of trace class operators acting on the bipartite Hilbert space $\H_{AB}$~\cite{chen_2003, rudolph_2003-1, rudolph_2003-2, rudolph_2004, rudolph_2005, horodecki_2006}. Here, we assume that $\H_A$ and $\H_B$ are isomorphic. A quick and painless way to define 
reshuffling in a rigorous manner while encompassing the case of infinite-dimensional Hilbert spaces is to think of it as an isometry on the Hilbert--Schmidt class $\HS(\H_{AB}) \supseteq \T(\H_{AB})$, defined on an arbitrary $X = \sum_{i,j,k,l} X_{ij,kl} \ketbraa{i}{k}_A\otimes \ketbraa{j}{l}_B$ (here $\{\ket{i}_A\}_{i\in \NN}$ and $\{\ket{j}_B\}_{j\in \NN}$ are two fixed local orthonormal bases, of the same cardinality because $\H_A\simeq \H_B$) by
\bb\label{eq:reshuffled}
X^\R \coloneqq \sum_{i,j,k,l} X_{ij,kl} \ketbraa{i}{j}_A\otimes \ketbraa{k}{l}_B ,
\ee
where the sum on the right-hand side converges in Hilbert--Schmidt norm because $\sum_{i,j,k,l} | X_{ij,kl}|^2 = \|X\|_{2}^2 <\infty$. A little thought reveals that $X\mapsto X^\R$ is in fact an isometry, so that
\bb
\Tr \left[ (X^\R)^\dag Y^\R \right] = \< X^\R, Y^\R \> = \<X,Y\> = \Tr \left[X^\dag Y\right]\qquad \forall X,Y\in \HS(\H_{AB}) .
\label{reshuffling_isometry}
\ee
The importance of the reshuffling operation stems from the fact that it maps separable states to trace class (rather than simply Hilbert--Schmidt) operators, and moreover~\cite{chen_2003, rudolph_2003-1, rudolph_2003-2, rudolph_2004, rudolph_2005, horodecki_2006}
\bb
\norm{\sigma^\R}{1} \leq 1 \qquad \forall \sigma\in \SEP .
\label{reshuffling_criterion}
\ee
Note that $\sigma^\R$ is not, in general, a self-adjoint operator, unlike the partially transposed state $\sigma^\Gamma$. However, $\big\|(\cdot)^\R\big\|_1$ can still be thought of as a gauge in the space of self-adjoint trace class operators. This is the gauge we are considering in this section.

To verify~\eqref{reshuffling_criterion}, we start by observing that for every two pure states $\ket{\psi} = \sum_i \psi_i \ket{i}\in \H_A$ and $\ket{\phi} = \sum_j \phi_j \ket{j} \in \H_B$ it holds that
\bb
\left(\ketbra{\psi}_A\otimes \ketbra{\phi}_B\right)^\R &= \left(\sum_{i,j,k,l} \psi_i\psi_k^* \phi_j \phi_l^*\, \ketbraa{i}{k}_A\otimes \ketbraa{j}{l}_B\right)^\R \\
&= \sum_{i,j,k,l} \psi_i\psi_k^* \phi_j \phi_l^*\, \ketbraa{i}{j}_A\otimes \ketbraa{k}{l}_B \\
&= \ketbraa{\psi}{\phi^*}_A \otimes \ketbraa{\psi^*}{\phi}_B .
\ee
Thus
\bb
\sigma^\R &= \left( \int \ketbra{\psi}_A \otimes \ketbra{\phi}_B\, \mathrm{d}\mu(\psi,\phi) \right)^\R \\
&= \int \left( \ketbra{\psi}_A \otimes \ketbra{\phi}_B \right)^\R \mathrm{d}\mu(\psi,\phi) \\
&= \int \ketbraa{\psi}{\phi^*}_A \otimes \ketbraa{\psi^*}{\phi}_B\, \mathrm{d}\mu(\psi,\phi) ,
\ee
where the identity in the second line follows from the fact that the 
integral~\eqref{separable} that describes the separable decomposition of $\sigma$, intended as a Bochner integral in the space of Hilbert--Schmidt operators, 
commutes with any continuous linear operator --- the reshuffling, being an isometry, is automatically continuous. Finally,
\bb
\norm{\sigma^\R}{1} = \norm{\int \ketbraa{\psi}{\phi^*}_A \otimes \ketbraa{\psi^*}{\phi}_B\, \mathrm{d}\mu(\psi,\phi)}{1} \leq \sup_{\ket{\psi},\ket{\phi}} \norm{\ketbraa{\psi}{\phi^*}_A \otimes \ketbraa{\psi^*}{\phi}_B}{1} = 1 .
\ee

We are now ready to state the following:

\begin{corollary}[Reshuffling lower bound]\label{cor:reshuffled_lower_bound}
Given an arbitrary bipartite quantum state $\rho_{AB}$, we have that
\bb
E_{c,\, \LOCC} (\rho_{AB}) \geq&\
E_{c,\, \NE}(\rho_{AB}) \geq \log N^\R_\tau(\rho_{AB}) ,
\ee
where the \deff{tempered reshuffled negativity} is defined by
\bb
N^\R_\tau(\rho_{AB}) \coloneqq&\ \max \lset \<W,\rho\> \bar \norm{W^\R}{\infty}\leq 1,\; \<W,\rho\> = \norm{W}{\infty} \rset ,
\label{tempered_reshuffled_negativity}
\ee
and the maximisation is over self-adjoint operators $W$.
\end{corollary}

\begin{remark}
The tempered reshuffled negativity, exactly as the tempered negativity introduced in~\cite{lami_2021-1}, can be computed efficiently via a semidefinite program. To see that this is the case, it suffices to reformulate slightly~\eqref{tempered_reshuffled_negativity} as
\begin{equation}\begin{aligned}
    N^\R_\tau(\rho_{AB}) = \max\lsetr \<W,\rho\> \barr \begin{pmatrix} \id & W^\R \\ (W^\R)^\dag & \id \end{pmatrix} \geq 0,\; -\<W,\rho\> \id \leq W \leq \<W,\rho\> \id \rsetr .
\end{aligned}\end{equation}
\end{remark}

\begin{proof}[Proof of Corollary~\ref{cor:reshuffled_lower_bound}]
We intend to apply Corollary~\ref{cor:single_letter_gamma_bound} with the choices $\norm{\cdot}{\mu} = \norm{\cdot}{\SEP}$ and $\norm{\cdot}{\gamma} = \norm{(\cdot)^\R}{1}$ (as mentioned, we think of $\norm{\cdot}{\gamma}$ as a gauge defined on the space of self-adjoint trace class operators). First of all, let us verify that $\norm{(\cdot)^\R}{1} \leq \norm{\cdot}{\SEP}$. To this end, for any $\delta>0$ pick some operator $X\in \T_{\sa}(\H)$ and a decomposition $X=a\sigma_+ - b \sigma_-$ with $\sigma_\pm \in \SEP$ and $a+b \leq \norm{X}{\SEP} + \delta$; we have that
\bb
\norm{X^\R}{1} = \norm{a\sigma_+^\R - b \sigma_-^\R}{1} \leq a \norm{\sigma_+^\R}{1} + b \norm{\sigma_-^\R}{1} \leq a+b \leq \norm{X}{\SEP} + \delta\, .
\ee
Since this holds for arbitrary $\delta>0$, we conclude that indeed $\norm{X^\R}{1} \leq \norm{X}{\SEP}$, as claimed.

We now compute the dual 
to $\norm{\cdot}{\gamma} = \norm{(\cdot)^\R}{1}$. Using~\eqref{reshuffling_isometry} yields that
\bb
\norm{Y}{\gamma}^\circ &= \sup\lset \<X,Y\> \bar X=X^\dag,\ \norm{X}{\gamma} \leq 1 \rset \\
&= \sup\lset \<X^\R,Y^\R\> \bar X=X^\dag,\ \norm{X^\R}{1} \leq 1 \rset \\
&= \norm{Y^\R}{\infty} .
\label{reshuffling_dual_norm}
\ee
Here, 
$X$ and $Y$ are understood to be of trace class. The equality in the last line is non-trivial. On the one hand, by H\"older's inequality it is clear that $\left|\<X^\R,Y^\R\>\right| \leq \norm{X^\R}{1} \norm{Y^\R}{\infty} \leq \norm{Y^\R}{\infty}$. To justify the reverse inequality, consider two normalised vectors $\ket{\psi}$ and $\ket{\phi}$ such that
\begin{equation}
\braket{\psi|Y^\R|\phi} = \norm{Y^\R}{\infty} .
\end{equation}
Now, set
\begin{equation}
X \coloneqq \frac12 \left( \ketbraa{\phi}{\psi}^\R + \left(\ketbraa{\phi}{\psi}^\R\right)^\dag \right) .
\end{equation}
Clearly, $X=X^\dag$ is self-adjoint. Furthermore, by an explicit element-wise computation using the definition of the reshuffled operator in~\eqref{eq:reshuffled}, one can verify that
\begin{equation}
\Big(\big(T^\R\big)^\dag \Big)^\R = F\overline{T}F
\end{equation}
holds for all Hilbert--Schmidt operators $T$, with $F$ being the swap operator defined by $F \ket{\alpha}\ket{\beta} = \ket{\beta}\ket{\alpha}$ and the bar denoting complex conjugation. We then have that
\begin{equation}
\norm{X^\R}{1} = \frac12\norm{\ketbraa{\phi}{\psi} + F\overline{\ketbraa{\phi}{\psi}} F}{1} \leq \frac12\norm{\ketbraa{\phi}{\psi}}{1} + \frac12 \norm{F\overline{\ketbraa{\phi}{\psi}} F}{1} = \norm{\ketbraa{\phi}{\psi}}{1} = 1\, .
\end{equation}
Finally,
\begin{equation} \begin{aligned}
\<X,Y\> &= \frac12 \<\ketbraa{\phi}{\psi}^\R, Y\> + \frac12 \<\left(\ketbraa{\phi}{\psi}^\R\right)^\dag, Y\> \\
&= \Re \<\ketbraa{\phi}{\psi}^\R, Y\> \\
&= \Re \<\ketbraa{\phi}{\psi}, Y^\R\> \\
&= \Re \braket{\psi|Y^\R|\phi} \\
&= \norm{Y^\R}{\infty} .
\end{aligned} \end{equation}
This completes the proof of~\eqref{reshuffling_dual_norm}.

The above calculation reveals that, fortunately, $\norm{\cdot}{\gamma}^\circ$ 
is indeed a multiplicative gauge, i.e.
\bb
\norm{Y^{\otimes n}}{\gamma}^\circ = \norm{(Y^{\otimes n})^\R}{\infty} = \norm{(Y^\R)^{\otimes n}}{\infty} = \norm{Y^\R}{\infty}^{n} = \left(\norm{Y}{\gamma}^\circ \right)^n
\ee
for all $Y$. We are thus ready to apply Corollary~\ref{cor:single_letter_gamma_bound}, which yields
\bb
E_{c,\, \O_\mu}(\rho_{AB})\, &=\, r(\Phi_2\toOmu \rho)^{-1} \\
&\geq\, \frac{1}{\loginf[\SEP]{\Phi_2}} \, \log \sup \lset \<W,\rho\> \bar \norm{W^\R}{\infty} \leq 1,\; \< W, \rho \> = \norm{W}{\infty} \rset \\
&=\, \log N_\tau^\R(\rho_{AB}),
\ee
where we recalled that $\loginf[\SEP]{\Phi_2} = 1$.

To complete the proof it suffices to note that non-entangling operations are always contractive with respect to the gauge $\norm{\cdot}{\SEP}$, so that $E_{c,\, \LOCC} (\rho_{AB}) \geq E_{c,\, \NE}(\rho_{AB}) \geq E_{c,\, \O_\SEP}(\rho_{AB})$.
\end{proof}

It is an open problem to find examples of states for which the reshuffled tempered negativity yields a lower bound on the entanglement cost that is both non-trivial and better than other bounds known so far, but we will shortly see that it can match the bound obtained from partial transposition, allowing us to provide an alternative derivation of the irreversibility of entanglement theory.

We remark in passing that both of the gauges associated with partial transposition and reshuffling, namely, $\|(\cdot)^\Gamma\|_1$ and $\|(\cdot)^\R\|_1$, can be seen to be weak* lower semicontinuous (which is relevant e.g.\ in the context of Theorem~\ref{thm:oneshot_dilution}). For example, writing
\begin{equation}
\big\|X^\Gamma\big\|_1 = \sup\lset \<X^\Gamma,Y\> \bar Y\in \HS(\H_{AB}),\ \|Y\|_\infty\leq 1 \rset = \sup\lset \<X,Z\> \bar Z\in \HS(\H_{AB}),\ \big\|Z^\Gamma\big\|_\infty\leq 1 \rset
\end{equation}
shows that the function $X\mapsto \|X^\Gamma\|_1$ is a pointwise supremum of weak*-lower-semicontinuous functions, implying that it is itself weak* lower semicontinuous.

\subsection{Recovering the irreversibility of entanglement theory}

To establish the asymptotic irreversibility of a theory, it suffices to exemplify states $\rho, \rho'$ such that $r(\rho \toO \rho') \, r(\rho' \toO \rho) < 1$. We will first show how this can be done in the theory of entanglement, recovering the recent result of~\cite{lami_2021-1} by means of a more general class of examples.

For an integer $d\geq 3$, define the $d\times d$-dimensional state $\omega_d$ as
\begin{equation}\begin{aligned}
	\omega_d \coloneqq \frac{1}{d(d-1)} \sum_{i,j=1}^d \left(\proj{ii} - \ketbraa{ii}{jj}\right) = \frac{1}{d-1} (P_d - \Phi_d),
\end{aligned}\end{equation}
where $\Phi_d \coloneqq \frac1d \sum_{i,j=1}^d \ketbraa{ii}{jj}$ and $P_d\coloneqq 
\sum_{i=1}^d \ketbra{ii}$. Note that for $d=3$ the above state reproduces precisely the state $\omega_3$ from~\cite{lami_2021-1}.

We proceed to establish a slight generalisation of the irreversibility of entanglement shown in~\cite{lami_2021-1}.
\begin{proposition}\label{prop:entanglement_irrev}
For any $d \geq 3$, the manipulation of the state $\omega_d$ in the resource theory of entanglement is irreversible under all non-entangling operations. Specifically,
\begin{equation}\begin{aligned}
E_{c,\, \NE}(\omega_d) > E_{d,\, \NE}(\omega_d).
\end{aligned}\end{equation}
\end{proposition}

\begin{proof}
Applying Corollary~\ref{cor:single_letter_bounds}, we obtain
\begin{equation}\begin{aligned}
	r(\omega_d \toOx{g,\SEP} \Phi_2) \leq \frac{\log \norm{\omega_d}{g,\SEP}}{-\log \norm{\Phi_2}{\SEP}^\circ} = \log \frac{d}{d-1},
\end{aligned}\end{equation}
where we 
observed that since $\omega_d + \frac{1}{d-1} \Phi_d = \frac{1}{d-1} P_d$, it follows that $R^g_\SEP(\omega_d) \leq \frac{1}{d-1}$. Incidentally, equality holds in this latter estimate, 
 although we shall not need this observation here.
For the lower bound, we use the base gauge $\norm{\cdot}{\SEP}$ and aim to bound the transformation rate $r(\Phi_2 \toO \omega_d)$ through Corollary~\ref{cor:single_letter_gamma_bound}. 
To this end, set
\bb
W_d \coloneqq \alpha_d P_d - \beta_d \Phi_d ,
\ee
where
\bb
\alpha_d \coloneqq \left\{ \begin{array}{cc} 2 & d=3, \\[1ex] \frac{d}{d-2} & d\geq 4 , \end{array}\right.\qquad \beta_d \coloneqq \left\{ \begin{array}{cc} 3 & d=3, \\[1ex] \frac{2d}{d-2} & d\geq 4 . \end{array}\right.
\ee
We obtain that
\bb
\norm{W_d^\R}{\infty} = \norm{\alpha_d P_d^\R - \beta_d \Phi_d^\R}{\infty} = \norm{\alpha_d P_d - \frac{\beta_d}{d} \id}{\infty} = \max\left\{ \left|\alpha_d-\frac{\beta_d}{d}\right|,\, \frac{|\beta_d|}{d} \right\} = 1 ,
\ee
and moreover
\bb
\< W_d, \omega_d \> = \Tr \left[ \left( \alpha_d P_d - \beta_d \Phi_d \right) \frac{P_d - \Phi_d}{d-1} \right] = \alpha_d = \max\{|\alpha_d|,\, |\alpha_d-\beta_d|\} = \norm{W_d}{\infty} ;
\ee
putting all together, we see that
\bb
E_{c,\, \NE}(\omega_d) \ &\geq r(\Phi_2 \toOx{\SEP} \omega_d)^{-1} \ \geq \log \< W_d, \omega_d \> = \log \alpha_d \\
&> \log \frac{d}{d-1} \geq r(\omega_d \toOx{g,\SEP} \Phi_2) \geq E_{d,\, \NE}(\omega_d),
\ee
as claimed.
\end{proof}

\begin{remark}
The proof technique in~\cite{lami_2021-1} (cf.~\eqref{eq:temp_neg_lower_bound}) uses in a similar way the  gauge $\norm{(\cdot)^\Gamma}{1}$ based on the partial transpose, instead of that based on the reshuffling criterion employed here. For the special state $\omega_d$, both of these choices lead to equally tight (and optimal, as can be seen by following~\cite{lami_2021-1}) bounds.
\end{remark}

\section{Application: resource theory of magic}\label{sec:app_magic}

The free states in the resource theory of magic are the stabiliser states $\F \coloneqq \STAB$, composed of convex mixtures of pure states $\proj{\psi_U}$ generated as $\ket{\psi_U} = U \ket{0}^{\otimes n}$ through the application of all Clifford unitaries $U$. 
Any state which is not in $\STAB$ is called a magic state, and in particular the distillation of highly resourceful, pure magic states is a cornerstone of many fault-tolerant quantum computation subroutines~\cite{bravyi_2005}. The resource theory of magic was formalised in~\cite{veitch_2014}, where transformation rates were studied for the first time. The free operations of this theory are typically taken to be the stabiliser operations $\OO_{\rm STAB}$, built through Clifford gates, Pauli measurements, and preparations of ancillary states in the computational basis. The motivation for such a choice can be understood as the fact that the application of any operation in $\OO_{\rm STAB}$ can be efficiently simulated on a classical computer~\cite{gottesman_1998}, while operations beyond $\OO_{\rm STAB}$ may require the costly magic states for their implementation. However, recent results showed that even larger classes of operations admit efficient classical simulation algorithms~\cite{veitch_2012,pashayan_2015,seddon_2019,wang_2019-1,rall_2019,seddon_2021,heimendahl_2022}, which could suggest that characterising the manipulation of magic states under larger types of operations may be of interest. 

\subsection{The case of qutrits}

The case of $d$-dimensional magic theory (i.e.\ defined for systems composed of $n$ qu\emph{d}its, where $d$ is an odd prime) is amenable to a particularly convenient characterisation, owing to the fact that the discrete Wigner function~\cite{gross_2006} can be defined (see~\cite{veitch_2014,wang_2020} for an  overview). The set of states with a positive Wigner representation, denoted $\FF_W$, then forms a useful approximation to the stabiliser states ${\STAB}$, and crucially it holds that any stabiliser protocol $\Lambda$ satisfies $\Lambda[\FF_W] \subseteq \FF_W$~\cite{veitch_2014}, meaning that such protocols are also free operations with respect to the set $\FF_W$ of states with a positive Wigner function. We then use this choice of free states to study the asymptotic properties of magic state transformations.

Let us hereafter focus on the case $d=3$. For any $(a_1, a_2) \in \mathbb{Z}_3 \times \mathbb{Z}_3$, the Heisenberg--Weyl operators are defined as
\begin{equation}\begin{aligned}
	T_{(a_1, a_2)} \coloneqq \omega^{-2 a_1 a_2} 
	Z^{a_1} X^{a_2},
\end{aligned}\end{equation}
where $\omega \coloneqq e^{2 \pi i / 3}$, and $X$ and $Z$ are the clock and shift operators, respectively. These operators are used to define the phase space point operators
\begin{equation}\begin{aligned}
	A_{(0,0)} &\coloneqq \frac{1}{3} \sum_{a_1, a_2 = 0}^{2} T_{(a_1,a_2)},\\
	A_{(a_1,a_2)} &\coloneqq T_{(a_1,a_2)}^{\phantom{\dag}} A_{(0,0)} T_{(a_1,a_2)}^\dagger,
\end{aligned}\end{equation}
which then allow us to define the \deff{discrete Wigner function} $W_{a_1,a_2}$ as
\begin{equation}\begin{aligned}
	W_{a_1,a_2} (Y) \coloneqq \frac{1}{3} \< A_{(a_1,a_2)}, Y \>.
\end{aligned}\end{equation}
For any state $\rho$, the Wigner representation $\{W_{a_1,a_2}(\rho)\}_{(a_1, a_2) \in \mathbb{Z}_3 \times \mathbb{Z}_3}$ forms a quasi-probability distribution over $\mathbb{Z}_3 \times \mathbb{Z}_3$. The corresponding Wigner trace norm is given by
\begin{equation}\begin{aligned}
	\norm{Y}{W} \coloneqq \sum_{a_1, a_2 = 0}^{2} \left|W_{a_1,a_2}(Y)\right|,
\end{aligned}\end{equation}
and the quantity $\log \norm{\rho}{W}$ has been dubbed the \deff{mana} of a quantum state $\rho$~\cite{veitch_2014}. 
The free states are then defined as
\begin{equation}\begin{aligned}
	\FF_W \coloneqq& \lset \rho \in \D(\H) \bar W_{a_1,a_2} (\rho) \geq 0 \; \forall (a_1, a_2) \in \mathbb{Z}_3 \times \mathbb{Z}_3 \rset\\
	=& \lset \rho \in \D(\H) \bar \norm{\rho}{W} = 1 \rset.
\end{aligned}\end{equation}
The generalisation to many copies is straightforward: one defines the Heisenberg--Weyl operators as $T_{(a_{1,1},a_{1,2}) \oplus \cdots \oplus (a_{n,1}, a_{n,2})} \coloneqq T_{(a_{1,1},a_{1,2})} \otimes \cdots \otimes T_{(a_{n,1},a_{n,2})}$, and the definitions of the phase space point operators and Wigner function are extended analogously, e.g.
\begin{equation}\begin{aligned}
	W_{a_{1,1},a_{1,2},\ldots,a_{n,1},a_{n,2}} (Y) = \frac{1}{3^n} \< A_{(a_{1,1},a_{1,2})\oplus\cdots\oplus(a_{n,1},a_{n,2})}, Y \>.
\end{aligned}\end{equation}
Crucially, the Wigner trace norm is multiplicative: $\norm{X\otimes Y}{W} = \norm{X}{W} \norm{Y}{W}$ for any $X,Y$~\cite{veitch_2014}. Coupled with the fact that $\norm{X}{W} \leq \norm{X}{\STAB} \, \forall X$, this will allow us to employ the Wigner trace norm in the multiplicative norm bound on resource cost (Corollary~\ref{cor:single_letter_gamma_bound}).

In this resource theory, states such as the Strange state $\ket{S} = (\ket{1} - \ket{2})/\sqrt{2}$ and the Norell state $\ket{N} = (-\ket{0} + 2\ket{1} - \ket{2})/\sqrt{6}$ have received attention as `maximally magical' states. In particular, Ref.~\cite{veitch_2014} noticed that both of these states maximise the Wigner trace norm $\norm{\cdot}{W}$ among all qutrit states, and raised the question of whether the states are asymptotically equivalent, that is, whether $r(\proj{S} \toO \proj{N}) = r(\proj{N} \toO \proj{S}) = 1$. Ref.~\cite{wang_2020} showed that this is not the case by proving that $r(\ket{N} \toO \ket{S}) < 1$. However, this did not rule out the possibility that $\ket{S}$ can be transformed to $\ket{N}$ at a rate which makes the two state interconvertible in the asymptotic limit.
An important question therefore remained open: is this resource theory asymptotically reversible?

\subsection{Irreversibility of magic manipulation}

Instead of the conversion between $\ket{S}$ and $\ket{N}$, we consider the transformation between the Norell state $\ket{N}$ and the Hadamard `+' state $\ket{H_+}$, defined as the $+1$ eigenstate of the Hadamard gate
\begin{equation}\begin{aligned}
	H = \frac{1}{\sqrt{3}} \begin{pmatrix} 1 & 1 & 1\\ 1 & \omega & \omega^2 \\ 1 & \omega^2 & \omega \end{pmatrix}.
\end{aligned}\end{equation}

To upper bound the rate $r(\proj{H_+} \to \proj{N})$ we use~\eqref{eq:handy_upper_distillable} with the choice of norm $\norm{\cdot}{g,\F_W}$ based on the generalised robustness monotone, defined as in~\eqref{eq:norm_g_FF}. In this case the bound reduces to that found in Ref.~\cite{wang_2020}:
\begin{equation}\begin{aligned}\label{eq:magic_wang_upperbound}
	r(\proj{H_+} \toOx{\F_W} \proj{N}) \leq \frac{\log\left(R^g_{\F_W}(\proj{H_+})+1\right)}{-\log \norm{\proj{N}}{\F_{W}}^\circ} = \frac{\log (3-\sqrt{3})}{\log \frac{3}{2}} \approx 0.59,
\end{aligned}\end{equation}
where we used the known value $\norm{\proj{N}^{\otimes n}}{\F_{W}}^\circ = \left(\frac23\right)^n$~\cite{wang_2020}, whose multiplicativity implies that $\loginfd[\FF_W]{\proj{N}} = \log \norm{\proj{N}}{\F_{W}}^\circ$, and the known value of $R^g_{\F_W}(\proj{H_+})= 2-\sqrt3$~\cite{wang_2020}.
For the other direction, we would like to use our Corollary~\ref{cor:single_letter_gamma_bound} (see also~\eqref{eq:hand_lower_cost}). Introducing the \deff{tempered~mana} 
\begin{equation}\begin{aligned}
	\NegW(\rho) \coloneqq \inf_{\norm{Z}{W}\leq 1} D_\h^{0}(\rho \| Z) = \log \max \lset \< \rho, X\> \bar \norm{X}{W}^\circ \leq 1, \; \norm{X}{\infty} = \Tr \rho X \rset,
\end{aligned}\end{equation}
where $\norm{X}{W}^\circ = \max_{a_1, a_2} \, d \left|W_{a_1,a_2}(X)\right|$ is known as the Wigner spectral norm, we can compute $W_\tau(\proj{H_+}) = \log \frac{1}{3}(1+2\sqrt{3})$ (see Appendix~\ref{appendix:wigner} for details).
However, the na\"ive bound $\loginf[\F_W]{\rho} \leq \log \norm{\rho}{\F_W}$ now only gives
\begin{equation}\begin{aligned}\label{eq:onlygives}
	r(\proj{N} \toOx{\F_W} \proj{H_+}) \leq \frac{\log \norm{\proj{N}}{\F_W}}{W_\tau(\proj{H_+})} = \frac{\log 2}{\log (1+2\sqrt{3})-\log 3} \approx 1.74,
\end{aligned}\end{equation}
which altogether only lets us bound the rates as
\begin{equation}\begin{aligned}
	r(\proj{H_+} \toOx{\F_W} \proj{N}) \; r(\proj{N} \toOx{\F_W} \proj{H_+}) \lessapprox 1.02.
\end{aligned}\end{equation}
We can see that the bound is \emph{almost} good enough, but not quite sufficient to show irreversibility --- we will need a slightly stronger bound here.

The advantage of the Wigner-based approach over an optimisation with respect to stabiliser states is that we can more easily evaluate $\norm{\cdot}{\F_W}$ for many copies of states. For example, we compute (see Appendix~\ref{appendix:wigner})
\begin{equation}\begin{aligned}\label{eq:moredetails1}
	\norm{\proj{N}^{\otimes 2}}{\F_W} = \frac{11}{3} < 4 = \norm{\proj{N}}{\F_W}^2.
\end{aligned}\end{equation}
Then, we get a better bound on the regularisation of $\norm{\cdot}{\F_W}$ as
\begin{equation}\begin{aligned}
	\loginf[\F_W]{\rho} \leq \frac12 \log \norm{\rho^{\otimes 2}}{\F_W},
\end{aligned}\end{equation}
which yields
\begin{equation}\begin{aligned}\label{eq:moredetails2}
	r(\proj{N} \toOx{\F_W} \proj{H_+}) \leq \frac{\frac12 \log\norm{\proj{N}^{\otimes 2}}{\F_W}}{W_\tau(\proj{H_+})} = \frac{\frac12 \log \frac{11}{3}}{\log (1+2\sqrt{3})-\log 3} \approx 1.63.
\end{aligned}\end{equation}
Combining with the reverse bound as above, we obtain
\begin{equation}\begin{aligned}
	r(\proj{H_+} \toOx{\F_W} \proj{N}) \; r(\proj{N} \toOx{\F_W}  \proj{H_+}) \leq \frac{\log \frac{11}{3} \log \left(3-\sqrt{3}\right)}{\log \frac{9}{4} \log \frac{1+2\sqrt{3}}{3}} \leq 0.96.
\end{aligned}\end{equation}
Asymptotic irreversibility is therefore established. Summing up, what we have shown is as follows.
\begin{boxed}{white}
\begin{theorem}\label{thm:magic_irreversibility}
The resource theory of multi-qudit magic is asymptotically irreversible under any class of operations $\OO$ that preserves the set of states with positive Wigner function. Specifically, if $\sigma \in \W \Rightarrow \Lambda(\sigma) \in \W$ for all $\Lambda \in \OO$, then
\begin{equation}\begin{aligned}
   r\!\left(\proj{H_+} \toO \proj{N}\right) < r\!\left(\proj{N} \toO  \proj{H_+}\right)^{-1}.
\end{aligned}\end{equation}
\end{theorem}
\end{boxed}
The result applies not only to stabiliser protocols, but also to the more general classes of completely stabiliser-preserving~\cite{seddon_2019} and completely positive-Wigner-preserving maps~\cite{wang_2019-1}, both of which are strictly larger than stabiliser operations~\cite{heimendahl_2022}.

We note that the irreversibility shown here is stronger than the irreversibility found in entanglement theory: the manipulation of \emph{pure} entangled states is known to be reversible~\cite{bennett_1996-1}, and only for a class of rank-2 states could the irreversibility under non-entangling transformations be established~\cite{lami_2021-1}. Here in the theory of magic, on the contrary, we see that not even pure states can be reversibly manipulated. 

We also stress that previously known bounds, including those based on the regularised relative entropy of magic $D_{\STAB}^\infty$~\cite{veitch_2014}, are not strong enough to show the irreversibility revealed by our results. To see this, note that~\cite{wang_2020}
\begin{equation}\begin{aligned}
    D_{\STAB}^\infty (\proj{H_+}) = \log\left(3-\sqrt{3}\right), \qquad D_{\STAB}^\infty (\proj{N}) = \log\frac{3}{2},
\end{aligned}\end{equation}
from which we only get that~\cite{veitch_2014}
\begin{equation}\begin{aligned}\label{eq:relent_rate_bound}
r\!\left(\proj{N} \toOx{\max} \proj{H_+}\right) \leq \frac{\log\frac{3}{2}}{\log(3-\sqrt{3})} \approx 1.71.
\end{aligned}\end{equation}
This is weaker than our bound in~\eqref{eq:moredetails2} and, in particular, insufficient to establish Theorem~\ref{thm:magic_irreversibility} --- one can notice that \eqref{eq:relent_rate_bound} is just the inverse of Eq.~\eqref{eq:magic_wang_upperbound}.

\subsection{The case of qubits}

We have been unable to decisively show that irreversibility occurs also in the resource theory of magic for qubits. However, we present below some partial results and a conjecture whose validity would indeed rule out the reversibility of many-qubit magic.

The discrete Wigner function does not allow for a straightforward application to the qubit case that would recover the nice properties of the qudit resource theory based on $\W$. However, a function conceptually similar to the Wigner negativity is the so-called stabiliser norm, defined for an $n$-qubit operator $X$ as~\cite{campbell_2011}
\begin{equation}\begin{aligned}
\norm{X}{\mathcal{P}} \coloneqq \frac{1}{2^n} \sum_{P \in \mathcal{P}}  \left|\Tr(X\, P)\right| ,
\end{aligned}\end{equation}
where $\mathcal{P}$ denotes all $n$-qubit Pauli operators. Its dual norm can be obtained as
\begin{equation}\begin{aligned}
	\norm{X}{\mathcal{P}}^\circ = \max_{P \in \mathcal{P}} \left|\Tr(X\, P)\right|.
\end{aligned}\end{equation}
Both of the norms are not difficult to see to be multiplicative on tensor products. 

Importantly, $\sigma \in {\STAB} \Rightarrow \norm{\sigma}{\mathcal{P}} \leq 1$, which means that $	\norm{X}{\mathcal{P}} \leq \norm{X}{\STAB}$.
The stabiliser norm approach does not always yield a good approximation for the set of stabiliser states: the set of $n$-qubit states with $\norm{\sigma}{\mathcal{P}} \leq 1$ is much larger than the set $\STAB$ when $n>1$~\cite{rall_2019}. Nevertheless, we can use it to define the \deff{tempered stabiliser norm}
\begin{equation}\begin{aligned}
	P_\tau(\rho) \coloneqq \inf_{\norm{Z}{\mathcal{P}}\leq 1} D_\h^{0}(\rho \| Z) = \log \max \lset \< \rho , X\> \bar \norm{X}{\mathcal{P}}^\circ \leq 1, \; \norm{X}{\infty} = \< \rho ,X\> \rset
\end{aligned}\end{equation}
and use the multiplicativity of $\norm{\cdot}{\mathcal{P}}^\circ$ to apply Corollary~\ref{cor:single_letter_gamma_bound} and establish $P_\tau$ as a single-letter bound on transformation rates.

An important class of states in the characterisation of multi-qubit magic are the Hoggar states $\ket{\rm Hog}$ (see e.g.~\cite{andersson_2015,howard_2017,stacey_2019}). A Hoggar state is any 3-qubit state defined by taking a fiducial state, e.g.
\begin{equation}\begin{aligned}
\ket{\rm Hog_0} \propto (-1+2i, 1, 1, 1, 1, 1, 1, 1)^T
\end{aligned}\end{equation}
and defining $\ket{\rm Hog}$ as belonging to the orbit of $\ket{\rm Hog_0}$ under the Pauli group, i.e.\ $\ket{\rm Hog} = P \ket{\rm Hog_0}$ for some $P \in \mathcal{P}$.
These states enjoy a very strong symmetry: it holds that
\begin{equation}\begin{aligned}\label{eq:hoggar_sym}
	\left|\braket{ {\rm Hog} | P | {\rm Hog}}\right| = \frac13
\end{aligned}\end{equation}
for all non-trivial Paulis $P \in \mathcal{P} \setminus \{\id\}$. Importantly, each Hoggar state satisfies the property that its standard robustness equals its generalised robustness~\cite{howard_2017,takagi_2021}:\footnote{We note that in~\cite{howard_2017}, the name `robustness of magic' was used to refer to the base norm $\norm{\cdot}{\STAB}$ itself; here it stands for $R^s_{\STAB}(\rho) = \frac12 \left(\norm{\rho}{\STAB} - 1\right)$, for consistency with the terminology used in entanglement theory and other resources.}
\begin{equation}\begin{aligned}
	R^s_{\STAB} (\proj{\rm Hog}) + 1= R^g_{\STAB} (\proj{\rm Hog}) + 1= \frac{1}{\norm{\proj{\rm Hog}}{{\STAB}}^\circ} = \frac{12}{5}.
\end{aligned}\end{equation}
Since the type of symmetry in Eq.~\eqref{eq:hoggar_sym} imposes a similar, rather symmetric structure on the tensor products of Hoggar states, one could reasonably expect that this will also lead to an equality between the robustness measures for many copies of this state.

\begin{conjecture}\label{conj_hoggar}
The standard robustness equals the generalised robustness for many copies of Hoggar states:
\begin{equation}\begin{aligned}
	R^s_{\STAB} (\proj{\rm Hog}^{\otimes n}) = R^g_{\F_{\rm STAB}} (\proj{\rm Hog}^{\otimes n}) \quad \forall n.
\end{aligned}\end{equation}
\end{conjecture}
Establishing the conjecture would then directly lead to the following.
\begin{corollary}
\textbf{If Conjecture~\ref{conj_hoggar} is true}, then the resource theory of many-qubit magic is asymptotically irreversible under stabiliser-preserving operations. Specifically, for any class of operations $\OO$ such that $\sigma \in \STAB \Rightarrow \Lambda(\sigma) \in \STAB \; \forall \Lambda \in \OO$, the single-qubit $\ket{T}$ state satisfies
\begin{equation}\begin{aligned}
	r\!\left(\proj{T} \toO \proj{\rm Hog}\right)  < r\!\left(\proj{\rm Hog} \toO \proj{T}\right)^{-1}.
\end{aligned}\end{equation}
\end{corollary}
\begin{proof}
The assumption that $R^s_{\STAB} (\proj{\rm Hog}^{\otimes n}) = R^g_{\STAB} (\proj{\rm Hog}^{\otimes n})$ tells us that $R^s_{\STAB}+1$ is sub-multiplicative for the Hoggar state. This gives
\begin{equation}\begin{aligned}
 \loginf[\STAB]{\proj{\rm Hog}} \leq \log ( 1 + R^s_{\STAB}(\proj{\rm Hog})) = \log \frac{12}{5}
\end{aligned}\end{equation}
where we used the known value of $R^s_{\STAB}(\proj{\rm Hog})$~\cite{howard_2017}. 
Computing $P_\tau(\proj{T}) = \frac{1+\sqrt{2}}{2}$ and applying Corollary~\ref{cor:single_letter_gamma_bound} then gives
\begin{equation}\begin{aligned}
	r(\proj{\rm Hog} \toOx{\STAB} \proj{T}) \leq \frac{\log \frac{12}{5}}{\log (1+\sqrt{2}) - 1} \approx 4.65.
\end{aligned}\end{equation}
On the other hand, using the fact that for any pure state of up to three qubits it holds that $\norm{\phi^{\otimes n}}{\STAB}^\circ = \left(\norm{\phi}{\STAB}^\circ\right)^n$~\cite{bravyi_2019}, we can upper bound the reverse transformation rate using the norm $\norm{\cdot}{g,{\STAB}}$ (which equals a bound previously shown in Ref.~\cite{seddon_2021}):
\begin{equation}\begin{aligned}
	r(\proj{T} \toOx{\sup} \proj{\rm Hog}) \leq \frac{\log \left(R^g_{\STAB}(\proj{T})+1\right)}{-\log \norm{\proj{\rm Hog}}{\STAB}^\circ} = \frac{1 + \log(2-\sqrt{2})}{\log \frac{12}{5}} \approx 0.18,
\end{aligned}\end{equation}
where we used that $R^g_{\STAB}(\proj{T}) + 1 = 2(2-\sqrt{2})$~\cite{bravyi_2019,seddon_2021} and $\norm{\proj{\rm Hog}}{\STAB}^\circ = \frac{5}{12}$~\cite{takagi_2021}.
Altogether,
\begin{equation}\begin{aligned}
	r(\proj{T} \toOx{\STAB} \proj{\rm Hog}) \, r(\proj{\rm Hog} \toOx{\sup} \proj{T}) \leq \frac{1+\log (2-\sqrt{2})}{\log(1+\sqrt{2})-1} \approx 0.84.
\end{aligned}\end{equation}
This concludes the proof.
\end{proof}

Of course, Conjecture~\ref{conj_hoggar} does not appear to be a necessary requirement for the irreversibility of many-qubit magic, and one could also expect that a transformation other than $\proj{\rm Hog} \leftrightarrow \proj{T}$ could be used to establish the irreversibility of the theory. We only present Conjecture~\ref{conj_hoggar} as one possible way to approach this problem, motivated by the fact that the tempered stabiliser norm $P_\tau$ provides a strong bound for Hoggar state manipulation.

Even if the conjecture itself is not found to be true, we hope that our discussion in this section motivates further research into the asymptotic manipulation of non-stabiliser states, as well as, more generally, into what is perhaps the most intriguing question prompted by our framework: what is it that makes a particular resource theory irreversible?\footnote{We remark here that, following up on earlier conjectures~\cite{brandao_2008-1,brandao_2010}, several recent works showed that reversibility of quantum resources can be recovered in various relaxed settings~\cite{regula_2024-1,wang_2023,ganardi_2024,hayashi_stein,blurring}. However, we stress that this can only be accomplished under classes of protocols that are strictly larger than the free operations $\OO_{\max}$, and indeed one needs to go significantly beyond conventional free operations to have any hope of achieving entanglement reversibility in general~\cite{lami_2021-1}. It is therefore still of significant practical and conceptual interest to understand which resources can be reversibly manipulated without such concessions.}

\bigskip
\noindent \textbf{Acknowledgments.} BR acknowledges the support of the Japan Society for the Promotion of Science (JSPS) KAKENHI Grant No.\ 22KF0067.  LL acknowledges support from the Alexander von Humboldt Foundation.

\newcommand{\nocontentsline}[3]{}
\let\addcontentsline=\nocontentsline

\bibliographystyle{apsc}
\bibliography{main}


\appendix

\section{Details of Wigner function computations}\label{appendix:wigner}

We first show Eq.~\eqref{eq:moredetails1}, namely
\begin{equation}\begin{aligned}
	\norm{\proj{N}^{\otimes 2}}{\F_W} = \frac{11}{3} < 4 = \norm{\proj{N}}{\F_W}^2.
\end{aligned}\end{equation}

For the single copy of the state $\proj{N}$, we can bound $\norm{\proj{N}}{\FF_W} \leq \norm{\proj{N}}{\STAB} = 2$, where the latter value was obtained in~\cite{takagi_2021}. The opposite inequality follows from the fact that $\norm{\proj{N}}{\FF_W} = 1 + 2 R^s_{\FF_W} (\proj{N}) \geq 1 + 2 R^g_{\FF_W} (\proj{N}) = 2$, where the last equality was shown in~\cite{wang_2020}.

Recall that, given the Heisenberg-Weyl operators $T_{(a_1,a_2)}$ of a single-qutrit Hilbert space, the bipartite operators take the form $T_{(a_1, a_2) \oplus (b_1,b_2)} = T_{(a_1,a_2)} \otimes T_{(b_1,b_2)}$. The bipartite Wigner representation is then defined analogously as $W_{a_1,a_2,b_1,b_2}(\rho) = \frac{1}{9} \Tr \rho A_{(a_1, a_2) \oplus (b_1,b_2)}$, where $A_{(a_1, a_2) \oplus (b_1,b_2)} = T_{(a_1, a_2) \oplus (b_1,b_2)} A_{(0,0)\oplus(0,0)} T^\dagger_{(a_1, a_2) \oplus (b_1,b_2)}$ with $A_{(0,0)\oplus(0,0)} = \frac{1}{9} \sum_{a_1,a_2,b_1,b_2=0}^{2} T_{(a_1, a_2) \oplus (b_1,b_2)}$.

Consider now the state $\proj{N}^{\otimes 2}$, whose Wigner representation $W_{a_1,a_2,b_1,b_2}(\proj{N}^{\otimes 2})$ takes the following values:

\renewcommand{\arraystretch}{1.2}
\begin{center}
\begin{tabular}{l|ccccccccc}
\diagbox{$(a_1,a_2)$}{$(b_1,b_2)$} & $(0,0)$ & $(0,1)$ & $(0,2)$ & $(1,0)$ & $(1,1)$ & $(1,2)$ & $(2,0)$ & $(2,1)$ & $(2,2)$ \\\midrule
$(0,0)$                                                         & $\frac{1}{36}$    & $-\frac{1}{18}$   & $\frac{1}{36}$    & $-\frac{1}{36}$   & $-\frac{1}{36}$   & $-\frac{1}{36}$   & $-\frac{1}{36}$   & $-\frac{1}{36}$   & $-\frac{1}{36}$   \\
$(0,1)$                                                         & $-\frac{1}{18}$   & $\frac{1}{9}$     & $-\frac{1}{18}$   & $\frac{1}{18}$    & $\frac{1}{18}$    & $\frac{1}{18}$    & $\frac{1}{18}$    & $\frac{1}{18}$    & $\frac{1}{18}$    \\
$(0,2)$                                                         & $\frac{1}{36}$    & $-\frac{1}{18}$   & $\frac{1}{36}$    & $-\frac{1}{36}$   & $-\frac{1}{36}$   & $-\frac{1}{36}$   & $-\frac{1}{36}$   & $-\frac{1}{36}$   & $-\frac{1}{36}$   \\
$(1,0)$                                                         & $-\frac{1}{36}$   & $\frac{1}{18}$    & $-\frac{1}{36}$   & $\frac{1}{36}$    & $\frac{1}{36}$    & $\frac{1}{36}$    & $\frac{1}{36}$    & $\frac{1}{36}$    & $\frac{1}{36}$    \\
$(1,1)$                                                         & $-\frac{1}{36}$   & $\frac{1}{18}$    & $-\frac{1}{36}$   & $\frac{1}{36}$    & $\frac{1}{36}$    & $\frac{1}{36}$    & $\frac{1}{36}$    & $\frac{1}{36}$    & $\frac{1}{36}$    \\
$(1,2)$                                                         & $-\frac{1}{36}$   & $\frac{1}{18}$    & $-\frac{1}{36}$   & $\frac{1}{36}$    & $\frac{1}{36}$    & $\frac{1}{36}$    & $\frac{1}{36}$    & $\frac{1}{36}$    & $\frac{1}{36}$    \\
$(2,0)$                                                         & $-\frac{1}{36}$   & $\frac{1}{18}$    & $-\frac{1}{36}$   & $\frac{1}{36}$    & $\frac{1}{36}$    & $\frac{1}{36}$    & $\frac{1}{36}$    & $\frac{1}{36}$    & $\frac{1}{36}$    \\
$(2,1)$                                                         & $-\frac{1}{36}$   & $\frac{1}{18}$    & $-\frac{1}{36}$   & $\frac{1}{36}$    & $\frac{1}{36}$    & $\frac{1}{36}$    & $\frac{1}{36}$    & $\frac{1}{36}$    & $\frac{1}{36}$    \\
$(2,2)$                                                         & $-\frac{1}{36}$   & $\frac{1}{18}$    & $-\frac{1}{36}$   & $\frac{1}{36}$    & $\frac{1}{36}$    & $\frac{1}{36}$    & $\frac{1}{36}$    & $\frac{1}{36}$    & $\frac{1}{36}$   
\end{tabular}
\end{center}

Define $X_+$ as the operator with the following Wigner representation:

\begin{center}
\begin{tabular}{l|ccccccccc}
\diagbox{$(a_1,a_2)$}{$(b_1,b_2)$} & $(0,0)$ & $(0,1)$ & $(0,2)$ & $(1,0)$ & $(1,1)$ & $(1,2)$ & $(2,0)$ & $(2,1)$ & $(2,2)$ \\\midrule
$(0,0)$                                                         & $\frac{1}{36}$    & $0$       & $\frac{1}{36}$    & $0$       & $0$       & $0$       & $0$       & $0$       & $0$       \\
$(0,1)$                                                         & $0$       & $\frac{2}{9}$     & $0$       & $\frac{1}{12}$    & $\frac{1}{12}$    & $\frac{1}{12}$    & $\frac{1}{12}$    & $\frac{1}{12}$    & $\frac{1}{12}$    \\
$(0,2)$                                                         & $\frac{1}{36}$     & $0$       & $\frac{1}{36}$     & $0$       & $0$       & $0$       & $0$       & $0$       & $0$       \\
$(1,0)$                                                         & $0$       & $\frac{1}{12}$    & $0$       & $\frac{1}{36}$    & $\frac{1}{36}$    & $\frac{1}{36}$    & $\frac{1}{36}$    & $\frac{1}{36}$    & $\frac{1}{36}$    \\
$(1,1)$                                                         & $0$       & $\frac{1}{12}$    & $0$       & $\frac{1}{36}$    & $\frac{1}{36}$    & $\frac{1}{36}$    & $\frac{1}{36}$    & $\frac{1}{36}$    & $\frac{1}{36}$    \\
$(1,2)$                                                         & $0$       & $\frac{1}{12}$    & $0$       & $\frac{1}{36}$    & $\frac{1}{36}$    & $\frac{1}{36}$    & $\frac{1}{36}$    & $\frac{1}{36}$    & $\frac{1}{36}$    \\
$(2,0)$                                                         & $0$       & $\frac{1}{12}$    & $0$       & $\frac{1}{36}$    & $\frac{1}{36}$    & $\frac{1}{36}$    & $\frac{1}{36}$    & $\frac{1}{36}$    & $\frac{1}{36}$    \\
$(2,1)$                                                         & $0$       & $\frac{1}{12}$    & $0$       & $\frac{1}{36}$    & $\frac{1}{36}$    & $\frac{1}{36}$    & $\frac{1}{36}$    & $\frac{1}{36}$    & $\frac{1}{36}$    \\
$(2,2)$                                                         & $0$       & $\frac{1}{12}$    & $0$       & $\frac{1}{36}$    & $\frac{1}{36}$    & $\frac{1}{36}$    & $\frac{1}{36}$    & $\frac{1}{36}$    & $\frac{1}{36}$   
\end{tabular}
\end{center}

and the operator $X_-$ with the following Wigner representation:

\begin{center}
\begin{tabular}{l|ccccccccc}
\diagbox{$(a_1,a_2)$}{$(b_1,b_2)$} & $(0,0)$ & $(0,1)$ & $(0,2)$ & $(1,0)$ & $(1,1)$ & $(1,2)$ & $(2,0)$ & $(2,1)$ & $(2,2)$ \\\midrule
$(0,0)$                                                         & $0$       & $\frac{1}{18}$    & $0$       & $\frac{1}{36}$    & $\frac{1}{36}$    & $\frac{1}{36}$    & $\frac{1}{36}$    & $\frac{1}{36}$    & $\frac{1}{36}$    \\
$(0,1)$                                                         & $\frac{1}{18}$    & $\frac{1}{9}$     & $\frac{1}{18}$    & $\frac{1}{36}$    & $\frac{1}{36}$    & $\frac{1}{36}$    & $\frac{1}{36}$    & $\frac{1}{36}$    & $\frac{1}{36}$    \\
$(0,2)$                                                         & $0$       & $\frac{1}{18}$    & $0$       & $\frac{1}{36}$    & $\frac{1}{36}$    & $\frac{1}{36}$    & $\frac{1}{36}$    & $\frac{1}{36}$    & $\frac{1}{36}$    \\
$(1,0)$                                                         & $\frac{1}{36}$    & $\frac{1}{36}$    & $\frac{1}{36}$    & $0$       & $0$       & $0$       & $0$       & $0$       & $0$       \\
$(1,1)$                                                         & $\frac{1}{36}$    & $\frac{1}{36}$    & $\frac{1}{36}$    & $0$       & $0$       & $0$       & $0$       & $0$       & $0$       \\
$(1,2)$                                                         & $\frac{1}{36}$    & $\frac{1}{36}$    & $\frac{1}{36}$    & $0$       & $0$       & $0$       & $0$       & $0$       & $0$       \\
$(2,0)$                                                         & $\frac{1}{36}$    & $\frac{1}{36}$    & $\frac{1}{36}$    & $0$       & $0$       & $0$       & $0$       & $0$       & $0$       \\
$(2,1)$                                                         & $\frac{1}{36}$    & $\frac{1}{36}$    & $\frac{1}{36}$    & $0$       & $0$       & $0$       & $0$       & $0$       & $0$       \\
$(2,2)$                                                         & $\frac{1}{36}$    & $\frac{1}{36}$    & $\frac{1}{36}$    & $0$       & $0$       & $0$       & $0$       & $0$       & $0$      
\end{tabular}
\end{center}

It then holds that $\proj{N}^{\otimes 2} = X_+ - X_-$ by construction. Furthermore, defining
\begin{equation}\begin{aligned}
\ket{v_1} &\coloneqq \frac{1}{\sqrt{10}} (1,0,1,-1,-2,-1,1,0,1)^T\\
\ket{v_2} &\coloneqq \frac{1}{2} (0,1,0,-1,0,-1,0,1,0)^T\\
\ket{v_3} &\coloneqq \frac{1}{3} (1,1,1,1,1,1,1,1,1)^T\\
\ket{v_4} &\coloneqq \frac{1}{\sqrt{6}} (-1,0,1,-1,0,1,-1,0,1)^T\\
\ket{v_5} &\coloneqq \frac{1}{2\sqrt{3}} (0,-1,-2,1,0,-1,2,1,0)^T
\end{aligned}\end{equation}
one can verify that
\begin{equation}\begin{aligned}
X_+ &= \proj{N}^{\otimes 2} + \frac{5}{12} \proj{v_1} + \frac{5}{12}\proj{v_2} + \frac13 \proj{v_3} + \frac{1}{2} \proj{v_4} + \frac{1}{12} \proj{v_5}\\
X_- &= \frac{5}{12} \proj{v_1} + \frac{5}{12}\proj{v_2} + \frac13 \proj{v_3} + \frac{1}{2} \proj{v_4} + \frac{1}{12} \proj{v_5},
\end{aligned}\end{equation}
which shows that $X_\pm$ are both positive semidefinite. As $X_\pm$ have a positive Wigner representation by construction, we get
\begin{equation}\begin{aligned}
	\norm{\proj{N}^{\otimes 2}}{\FF_W} \leq \Tr X_+ + \Tr X_- = \frac{11}{3}.
\end{aligned}\end{equation}
This value can be verified to be optimal by employing the dual form of $\norm{\cdot}{\FF_W}$, but we do not need this for our argument.

We now show how to bound the value of $W_\tau(\proj{H_+})$, which we used in Eqs.~\eqref{eq:onlygives} and~\eqref{eq:moredetails2}. Recall that
\begin{equation}\begin{aligned}
	W_\tau(\rho) = \log \max \lset \< \rho, X\> \bar \norm{X}{W}^\circ \leq 1, \; \norm{X}{\infty} = \Tr \rho X \rset,
\end{aligned}\end{equation}
where
\begin{equation}\begin{aligned}
	\norm{X}{W}^\circ = \max_{a_1, a_2} \, 3 \left|W_{a_1,a_2}(X)\right|.
\end{aligned}\end{equation}
Consider then the ansatz
\begin{equation}\begin{aligned}
	X = \frac{1+2\sqrt{3}}{3} \proj{H_+} - \frac{2\sqrt{3}-1}{3} \proj{H_-} - \frac{1}{3} \proj{H_i},
\end{aligned}\end{equation}
where $\ket{H_-}$ and $\ket{H_i}$ are the eigenvectors of the Hadamard gate corresponding to the eigenvalues $-1$ and $+i$, respectively.
Clearly, $\norm{X}{\infty} = \braket{H_+ | X | H_+} = \frac{1+2\sqrt{3}}{3}$. Computing the Wigner representation of $X$, we can see that it takes the form:
\begin{center}
\begin{tabular}{l|ccc}
\diagbox{$a_1$}{$a_2$\;} & $0$ & $1$ & $2$ \\\midrule
$0$ & $\frac{1}{3}$ & $\frac{1}{3}$ & $\frac{1}{3}$\\
$1$ & $\frac{1}{3}$ & $-\frac{1}{3}$ & $-\frac{1}{3}$\\
$2$ & $\frac{1}{3}$ & $-\frac{1}{3}$ & $-\frac{1}{3}$
\end{tabular}
\end{center}
Thus $\norm{X}{W}^\circ = 1$, and we conclude that
\begin{equation}\begin{aligned}
	W_\tau(\proj{H_+}) \geq \log \braket{H_+ | X | H_+} = \log \frac{1+2\sqrt{3}}{3}.
\end{aligned}\end{equation}
A numerical evaluation can be used to confirm that this value is in fact optimal.

\end{document}

%% file: def.tex
\pdfoutput=1
\usepackage{bm,bbm}%
\usepackage{braket}
\usepackage{amsthm,amsmath,amssymb}

\usepackage[utf8]{inputenc}

\usepackage{newpxtext,newpxmath}
\usepackage[dvipsnames,svgnames]{xcolor}
\usepackage{amsmath,amssymb}
\usepackage{etruscan}
\usepackage{tikz}
\usepackage{enumerate}
\usepackage{mathtools}
\usepackage[colorlinks=true, allcolors=MidnightBlue!70!black!70!TealBlue,pdftitle={Functional analytic insights into irreversibility of quantum resources}, pdfauthor={Bartosz Regula and Ludovico Lami}, linktoc=all]{hyperref}

\usepackage{array,booktabs,multirow}

\newtheorem{theorem}{Theorem}
\newtheorem{conjecture}[theorem]{Conjecture}
\newtheorem{proposition}[theorem]{Proposition}
\newtheorem{corollary}[theorem]{Corollary}
\newtheorem{definition}{Definition}
\newtheorem{lemma}[theorem]{Lemma}

\theoremstyle{definition}
\newtheorem*{remark}{Remark}

\newtheoremstyle{red}{}{}{\normalfont}{}{\color{red!80!black}\bfseries}{.}{ }{}
\theoremstyle{red}

\let\nc\newcommand

\nc{\NegW}{W_\tau}

\newcommand{\<}{\left\langle}
\renewcommand{\>}{\right\rangle}

\renewcommand{\bar}{\;\rule{0pt}{9.5pt}\right|\;}

\newcommand{\lset}{\left\{\left.}
\newcommand{\rset}{\right\}}
\newcommand{\lsetr}{\left\{\,}
\newcommand{\rsetr}{\right.\right\}}
\newcommand{\barr}{\;\rule{0pt}{9.5pt}\left|\;}

\DeclareMathOperator{\Tr}{Tr}
\DeclareMathOperator{\supp}{supp}

\nc{\Neg}{\textetr{n}}
\nc{\Rob}{\textetr{r}}
\nc{\h}{\textetr{h}}

\let\textleq\relax

\newcommand{\texteq}[1]{\stackrel{\mathclap{\scriptsize \mbox{#1}}}{=}}
\newcommand{\textleq}[1]{\stackrel{\mathclap{\scriptsize \mbox{#1}}}{\leq}}

\newcommand{\norm}[2]{\left\lVert#1\right\rVert_{\,#2}}
\newcommand{\proj}[1]{\ket{#1}\!\bra{#1}}

\newcommand{\ketbra}[1]{\ket{#1}\!\bra{#1}}
\newcommand{\ketbraa}[2]{\ket{#1}\!\bra{#2}}
\nc{\id}{\mathbbm{1}}

\newcommand{\HS}{\mathcal{HS}}

\newcommand{\B}{\mathcal{B}}

\newcommand{\D}{\mathcal{D}}
\renewcommand{\H}{\mathcal{H}}
\newcommand{\T}{\mathcal{T}}
\newcommand{\sa}{{\mathrm{sa}}}

\newcommand{\V}{\mathcal{V}}

\newcommand{\RR}{\mathbb{R}}
\newcommand{\CC}{\mathbb{C}}

\newcommand{\NN}{\mathbb{N}}

\newcommand{\R}{\mathrm{\scriptstyle R}}
\newcommand{\sumno}{\sum\nolimits}

\nc{\PPTP}{\mathrm{PPTP}}
\nc{\SEPP}{\mathrm{SEPP}}
\nc{\SP}{\mathrm{KP}}
\nc{\FP}{\mathrm{FP}}
\nc{\CPTP}{{\mathrm{CPTP}}}
\nc{\CP}{{\mathrm{CP}}}
\let\wt\widetilde

\nc{\Op}{\mathcal{O}}

\nc{\idc}{\mathrm{id}}
\nc{\ve}{\varepsilon}

\nc{\Omax}{\O_{\mathrm{max}}}

\usepackage{stackengine,moresize}
\stackMath
\newcommand\xxrightarrow[2][]{\mathrel{%
  \setbox2=\hbox{\stackon{\scriptstyle#1}{\scriptstyle#2}}%
  \stackunder[0pt]{%
    \xrightarrow{\makebox[\dimexpr\wd2\relax]{$\scriptstyle#2$}}%
  }{%
   \scriptstyle#1%
  }%
}}

\newcommand{\toO}{\xxrightarrow[\vphantom{a}\smash{{\raisebox{-.5pt}{\ssmall{$\O$}}}}]{}}
\newcommand{\toOmu}{\xxrightarrow[\vphantom{a}\smash{{\raisebox{-.5pt}{\ssmall{$\O_\mu$}}}}]{}}
\newcommand{\toOx}[1]{\xxrightarrow[\vphantom{a}\smash{{\raisebox{-.5pt}{\ssmall{$\O_{#1}$}}}}]{}}

\newcommand{\floor}[1]{\left\lfloor #1 \right\rfloor}
\newcommand{\ceil}[1]{\left\lceil #1 \right\rceil}

\usepackage[most,breakable]{tcolorbox}
\renewenvironment{boxed}[1]%
	{\expandafter\ifstrequal\expandafter{#1}{orange}{\begin{tcolorbox}[colback=red!15,colframe=orange!15,breakable,enhanced]}{\begin{tcolorbox}[colback=white,colframe=gray!10,breakable,enhanced]}}%
	{\end{tcolorbox}}
	
\newcommand{\bb}{\begin{equation}\begin{aligned}\hspace{0pt}}
\newcommand{\bbb}{\begin{equation*}\begin{aligned}}
\newcommand{\ee}{\end{aligned}\end{equation}}
\newcommand{\eee}{\end{aligned}\end{equation*}}

\usepackage{diagbox}

\newcommand{\OO}{\mathbb{O}}
\newcommand{\FF}{\mathbb{F}}
\let\F\FF
\let\O\OO

\nc{\STAB}{{\FF_{\rm STAB}}}
\nc{\SEP}{{\FF_{\rm SEP}}}
\nc{\PPT}{{\FF_{\rm PPT}}}
\nc{\W}{{\FF_{\rm W}}}

\nc{\LOCC}{\OO_{\mathrm{LOCC}}}
\nc{\NE}{\OO_{\mathrm{NE}}}
\nc{\SEPO}{\OO_{\mathrm{SEP}}}
\nc{\PPTO}{\OO_{\mathrm{PPT}}}

\usepackage{pgfplots}
\pgfplotsset{width=11cm,compat=1.9}

\makeatletter
\def\l@subsection#1#2{}
\def\l@subsubsection#1#2{}
\makeatother

\makeatletter
\patchcmd{\@ssect@ltx}
    {\addcontentsline{toc}{#1}{\protect\numberline{}#8}}
    {}
    {}
    {}
\makeatother

\let\epsilon\varepsilon

\newcommand{\loginf}[2][\mu]{L_{#1,\infty}(#2)}
\newcommand{\loginfd}[2][\mu]{L_{#1,\infty}^\circ(#2)}